\documentclass{llncs}
\usepackage{fullpage}
\usepackage{graphicx} 
\usepackage{xcolor}
\usepackage[T1]{fontenc}
\usepackage{amssymb}
\usepackage{mathtools}

\usepackage{physics}
\usepackage{tabularx}
\usepackage{hyperref}
\usepackage{cleveref}
\usepackage{amsmath}

\newcounter{protocol}
\newenvironment{protocol}[1]
  {\par\addvspace{\topsep}
   \noindent
   \tabularx{\linewidth}{@{} X @{}}
    \hline
    \vspace{-2mm}
    \refstepcounter{protocol}\textbf{Protocol \theprotocol} #1 \\
    \hline}
  {\\
  \hline
   \endtabularx
   \par\addvspace{\topsep}}

\usepackage{cite}

\newcommand{\negl}{\textsf{negl}}

\newcommand{\N}{\mathbb{N}}
\newcommand{\R}{\mathfrak{R}}

\newcommand{\V}{Ver}
\newcommand{\Vt}{\mathcal{V}}

\newcommand{\id}{\text{Id}}
\newtheorem{Proposition}[proposition]{Proposition}
\newcommand{\tl}{\tilde l}

\newcommand{\F}{\mathbb{F}}

\newcommand{\tP}{\tilde P}

\title{On the Relativistic Zero Knowledge Quantum Proofs of Knowledge}
\institute{}
\author{Kaiyan Shi \inst{1,2}, Kaushik Chakraborty\inst{1}, Wen Yu Kon\inst{1}, Omar Amer\inst{1}, Marco Pistoia\inst{1}, Charles Lim\inst{1}}
\institute{Global Technology Applied Research, JPMorganChase \\
\and 
Department of Computer Science, University of Maryland}

\begin{document}

\maketitle

\begin{abstract}
We initiate the study of relativistic zero-knowledge quantum proof of knowledge systems with classical communication, formally defining a number of useful concepts and constructing appropriate knowledge extractors for all the existing protocols in the relativistic setting which satisfy a weaker variant of the special soundness property due to Unruh (EUROCRYPT 2012). We show that there exists quantum proofs of knowledge with knowledge error $\frac{1}{2} + \textsf{negl}(\eta)$ for all relations in NP via a construction of such a system for the Hamiltonian cycle relation using a general relativistic commitment scheme exhibiting the fairly-binding property due to Fehr and Fillinger (EUROCRYPT 2016). We further show that one can construct quantum proof of knowledge extractors for proof systems which do not exhibit special soundness, and therefore require an extractor to rewind multiple times. We develop a new multi-prover quantum rewinding technique by combining ideas from monogamy of entanglement and gentle measurement lemmas that can break the quantum rewinding barrier. Finally, we prove a new bound on the impact of consecutive measurements and use it to significantly improve the soundness bound of some existing relativistic zero knowledge proof systems, such as the one due to Chailloux and Leverrier (EUROCRYPT 2017). 
\end{abstract}

\section{Introduction}
\label{sec:Introduction}
Multi-party interactive protocols involve multiple players attempting to jointly complete a cryptographic task without trusting each other. It is known that without introducing additional assumptions, some multi-party functionalities cannot be securely realized in the information theoretic setting, in which the adversary has unlimited classical or quantum computational power. 
One such assumption that can be adopted dates back to~\cite{ben2019multi}, where the communication between some players is prohibited to ensure information-theoretic security of the protocol.
However, enforcing the non-communication between players can be challenging in realistic scenarios.

One method of enforcing this assumption is to exploit the fact that information cannot be transmitted faster than the speed of light.
This can be implemented in practice by ensuring that the parties are far enough such that no information can be exchanged between them before the protocol ends, i.e. having the parties be distance $d>ct$ apart for protocol run time $t$, where $c$ is the speed of light (more formally that the parties remain spacelike separated during the protocol).
We call protocols that make use of this property relativistic multi-party interactive protocols, where the relativistic setting was pioneered by Kent in 1999~\cite{kent2005secure} in an effort to bypass the no-go theorems for quantum bit-commitment~\cite{mayers1997unconditionally,lo1997quantum}. Later, Crepeau et al.~\cite{crepeau2011two} revisited the original idea of Ref.~\cite{ben2019multi} and developed a relativistic bit commitment protocol with both information-theoretically secure hiding and binding properties. In such a relativistic bit commitment protocol, one player commits the bit to an agent, and the other player reveals the bit when required. However, for these relativistic protocols, the distance between the parties depends highly on the duration of the protocol. Lunghi et al.~\cite{lunghi2015practical} devised a multi-round variant of the bit commitment protocol in Ref.~\cite{crepeau2011two} to overcome this limitation, which allows the duration of the bit commitment scheme to be extended by increasing the number of rounds.
While this protocol was shown to be secure against classical adversaries~\cite{chakraborty2015arbitrarily,fehr2016composition,lunghi2015practical}, its security against quantum adversaries remains an open problem. 
Later, in Ref.~\cite{Verbanis_2016}, the authors implemented a 24-hour-long bit commitment by keeping the players 8km apart. 

Zero-knowledge protocols (ZKPs) form an important class of multi-party interactive protocols, where the goal is to have provers prove a statement to verifiers without providing the verifiers with additional information.
Recently, Chailloux and Leverrier~\cite{chailloux2017relativistic} designed a perfect zero-knowledge protocols in the relativistic setting by adopting Blum's ZKP~\cite{blum1986prove} and combining it with the relativistic bit commitment protocol from Ref.~\cite{crepeau2011two}. The construction in Ref.~\cite{chailloux2017relativistic} was for the Hamiltonian cycle, with two distant provers responding to challenges by the verifiers. However, the protocol is limited by its high communication costs. Later, Crepeau et al.~\cite{crepeau2019practical} proposed a relativistic ZKP protocol for graph 3-coloring, using a completely different approach, with a very low communication and computation overhead. They refer to this approach as the \emph{unveil-via-commit principle} which is inspired from the double-spending detection mechanism of the untraceable electronic cash of Chaum, Fiat and Naor~\cite{chaum1990untraceable}. 
The practicality of this protocol was demonstrated in a follow-up experiment with a real-time implementation and the provers being $50$ meters apart~\cite{alikhani2021experimental}.
This can be useful in applications such as ATM-card authentication, as described by Brassard~\cite{brassard2021relativity}.
However, to claim security against quantum adversaries, the authors~\cite{crepeau2019practical} required three provers instead of two. Recently, Crepeau and Stuart~\cite{crepeau2023zero} again used the homomorphic property of the $\F_Q$-bit commitment scheme and proposed an efficient construction for the subset-sum and SAT problem. Although this protocol's communication and computation overhead is slightly larger than that for graph 3-coloring~\cite{crepeau2019practical}, this recent construction requires only two provers to be secure against quantum adversaries.

The security analysis of existing relativistic ZKPs are limited to proving the soundness and zero-knowledge properties.
However, the proof of soundness against quantum adversaries does not necessarily imply a quantum-proof of knowledge, with some known counterexamples \cite{ambainis2014quantum}.  
To the best of our knowledge, there has been no study of the quantum proofs of knowledge (QPoK) of these relativistic zero-knowledge protocols. In a nutshell, this paper addresses the following question.

\begin{equation*}
\textit{Are the existing relativistic zero-knowledge protocols QPoK?}
\end{equation*}

\medskip \noindent \textit{Remark:} We note that there are some subtleties one must navigate in attempting to answer this question using tools from the single prover setting. For example, while it is natural to attempt to use strong tools such as the rewinding lemmas and state repair procedures of \cite{chiesa2022post} to deal with the rewinding of the quantum provers, it is not straightforwad to do so in our setting. Indeed, in the relativistic setting, where unconditional commitment schemes are possible, using such tools as currently presented would require limiting the computational power of the provers, whereas our analysis is unconditional. Similarly, while some of the extractors we build follow closely the blueprint for extracting from $\Sigma$-protocols with both strict and special soundness given by Unruh \cite{unruh2012quantum}, the structure of $\Sigma$ protocols differs in the relativistic setting and requires dedicated modeling. Addressing this difference allows us to relax some of the requirements (namely, strict soundness) necessary for Unruh's original work. Intuitively, the fact that our spacelike separated provers lack information regarding the challenges sent to the other prover restricts the ability of the prover to devise multiple accepting answers, even though multiple accepting answers exist.

\medskip \noindent \textbf{Organization of the Paper.} In \Cref{sec:our-results} we give an informal overview of the results of this paper. In \Cref{sec:prelims} we recall a collection of useful definitions and results from the literature, which we use throughout the paper to prove our results. In \Cref{sec:QPoKs-Defs} we introduce useful definitions of interaction in the relativistic setting. We then use them to show in \Cref{sec:QPoKswithSpec} that a class of specially-sound $\Sigma$-protocols from the literature are Quantum Proofs of Knowledge. In \Cref{sec:QPoKs-NP} we generalize to show that one can construct QPoKs for all of NP using a relativistic bit commitment scheme with a binding property due to Fehr and Fillinger \cite{fehr2016composition}. In \Cref{sec:QPOKs-Symm} we leverage tools from the literature on symmetric games, monogamy of entanglement, and gentle measurements in order to show that proof systems with symmetric challenge spaces can be QPoKs, even if they lack the special soundness property and therfore require multiple rewindings to extract a prover. Finally, in \Cref{sec:improv-sound} we give a new lemma regarding the outcome of two consecutive measurements, and use it to prove improved lower bounds on the soundness error of some relativistic zero knowledge proof systems from the literature. 

\subsection{Our Results}
\label{sec:our-results}
\subsubsection{Rewinding Oracles in the Relativistic Setting:}

For the analysis of quantum proofs of knowledge, we first define interactive machines in a relativistic setting. Though these interactive machines are similar to multi-prover interactive proof systems, to the best of our knowledge, a formal description of such machines in a relativistic setting remains elusive. 

Taking advantage of the existing results on rewinding oracles in the non-relativistic setting, we further define rewinding oracles in the relativistic setting.
The only difference between the rewinding oracles in relativistic and non-relativistic settings is that, due to their spacelike separation, the provers can only apply quantum operations in tensor product form on a quantum system that is shared between the provers in the relativistic setting. 

\subsubsection{Relativistic $\Sigma$-Protocol:} 
Similar to $\Sigma$-protocols in the non-relativistic setting~\cite{unruh2012quantum}, we also define $\Sigma$-protocols in the relativistic setting.

\begin{definition} [Relativistic $\Sigma$-protocol] \label{def:rel-sigma-protocol}
    A $2$-prover interactive proof system $\Pi = (P_1,P_2,V)$ where the provers $P_1,P_2$ are separated by distance $dist$, is called a relativistic $\Sigma$-protocol if and only if the interactions consist of four messages $(rand, com) , (ch, resp)$ sent sequentially by $V,P_1,V,P_2$, where $rand$ and $ch$ are chosen uniformly at random by the verifier from the field $\F_{Q(\eta x)}$ and a challenge space $C_{\eta x}$ respectively where each depends only on the security parameter $\eta$ and input instance $x$, and from which we can efficiently sample. We allow the two provers hold some shared randomness and note that the Commit and Unveil phases occur simultaneously. The protocol follows as below. ~\footnote{In general practice, we have two verifiers $V_1$ and $V_2$, each positioned to corresponding provers.}
    
   \medskip\noindent \textsf{Commit Phase.}
    \begin{enumerate}
    \vspace{-1mm}
        \item $V$ sends $rand$ to $P_1$ at time $T_1$.
        \item $P_1$ replies with the commitment $com$.
        \item $V$ receives the commitment $com$ from $P_1$ at time $T_2$.
    \end{enumerate}
    \vspace{-2mm}
    \medskip\noindent \textsf{Unveil Phase.}
    \begin{enumerate}
        \vspace{-1mm}
        \item $V$ sends the challenge $ch$ to $P_2$ at $T_1$.
        \item $P_2$ replies with $resp$.
        \item $V$ receives the response $resp$ at time $T_3$.
    \end{enumerate}
    \vspace{-2mm}
    \medskip\noindent \textsf{Verification Phase.}
    \begin{enumerate}
    \item $V$ aborts the protocol if $(T_3 - T_1) \geq \frac{dist}{c}$ or $(T_2 - T_1) \geq \frac{dist}{c}$, where $c$ is the speed of light.
    \item $V$ conducts the verification via some deterministic efficient computation based on $rand, ch, com$ and $resp.$
    \end{enumerate}
\end{definition}

We note that there is a clear equivalence between relativistic $\Sigma$-protocols where the verifier does not abort (i.e., where the provers respond in a relativisticly constraining time period), and protocols in which the provers are assumed not to communicate within rounds. As such, we will often treat the two cases equivalently for simplicity.

We re-define a probabilistic variant of the special soundness property in the relativistic setting. 

\begin{definition} [$\delta_{SS}$-Special soundness~\cite{unruh2012quantum}] \label{defn:dss-special-soundness}
A relativistic $\Sigma$-protocol $(P_1,P_2,V)$ for  relation $R$ is said to have $\delta_{SS}$-special soundness if and only if there exists a deterministic polynomial-time algorithm $K_0$ (the special extractor) such that the following holds: For any two accepting instances $(rand, com, ch, resp)$ and $(rand, com, ch', resp')$ for $x$ such that $ch \neq ch'$ and $ch,ch' \in C_{\eta x}$, we have that $$w:=K_0(x,rand,com,ch,resp,ch',resp'),$$ satisfies $(x,w) \in R$ with probability at least $1-\delta_{SS}$. Formally speaking, for all $com$, $rand$,
\begin{align*}
    &\Pr_{ch,resp,ch',resp'}[(x,w) \in R, w \leftarrow K_0(x,rand,com,ch,resp,ch',resp')|rand,com,acc] \\
    &\geq (1- \delta_{SS}),
\end{align*}
where $acc$ represents the event that two conversations are acceptable. 
\end{definition}

Above and elsewhere in the text, $\delta_{SS}$ will refer to the special soundness parameter and bears no relation to the commonly used Kronecker delta function.

\subsubsection{Quantum Proofs of Knowledge for Relativistic $\Sigma$-Protocols with $\delta_{SS}$-Special Soundness}
We show the viability of relativistic $\Sigma$-protocols which use the $\F_Q$-bit commitment scheme as a subroutine and exhibit the special soundness property as quantum proofs of knowledge. Interestingly, the protocols proposed in~\cite{chailloux2017relativistic}, and~\cite{crepeau2023zero} fall into this category, which we prove in \Cref{app:ex_sigma}. We arrive at the below theorem using our model for interactive provers and results on consecutive measurements from \cite{chailloux2017relativistic}.

\begin{theorem}[Informal]
\label{thm:rzkpFQ}
If a relativistic $\Sigma$-protocol has the $\delta_{SS}$-special soundness property then, it has a quantum proof of knowledge with knowledge error $\frac{1}{2} + \negl(\eta)$.
\end{theorem}

We generalize the results to obtain special soundness for a $\Sigma$-protocol for Hamiltonian cycle using any relativistic bit commitment scheme which exhibits the fair-binding property as define in \Cref{def:fairly-binding}. Using a commitment scheme with this fairly-binding property, we can construct a proof system for Hamiltonian cycle which we can show still satisfies special soundness, from which we obtain the following corollary.

\begin{corollary} [Relativistic QPoK for all languages in NP (Informal)]
    Let $R$ be an NP-relation. Assuming the existence of $\epsilon_{FB}$-fair-binding bit commitment scheme with $\epsilon_{FB} < \frac{1}{\textsf{poly}(\eta)}$, there is a zero-knowledge relativistic QPoK for $R$ with knowledge error $\frac{1}{2} + \negl(\eta)$.
\end{corollary}

\subsubsection{Quantum Proofs of Knowledge for Symmetric Relativistic Zero-Knowledge Protocols:}
We also consider relativistic zero-knowledge protocol without the special soundness properties. Due to the absence of the special soundness property, the knowledge extractor may need to perform rewinding an arbitrary number of times. If we analyze the knowledge error for such extractors using consecutive measurement theorems like our Theorem \ref{thm:consec_inf}, or the result obtained by Don et al.~\cite{don2019security}, it will lead to a trivial bound. In \cite{chiesa2022post}, Chiesa et al. addresses the same issue in the non-relativistic setting by introducing a state repair procedure. However, it is unclear whether we can apply their technique directly in the relativistic setting due to the unbounded computation power of the provers.

The analysis of such protocols can be approached differently. We note that the most general action that malicious provers can perform in two-prover relativistic ZKP involve sharing an entangled state and measuring each subsystem separately. This is in general more powerful than classical attacks without shared entanglement, resulting in a substantial gap in winning probabilities between classical and entangled provers. The famous CHSH game~\cite{bell1964einstein} is one such example where the gap is significant. Entangled strategies that have significantly higher winning probability than the optimal classical strategy typically requires the sharing of a maximally entangled state. Since quantum entanglement is monogamous, we may introduce a third prover to corroborate other provers' responses as a means to weaken entangled strategies. According to the monogamy of entanglement property, for any tripartite quantum state shared between three provers, having two maximally entangled provers would mean that neither can be entangled with the third prover \cite{terhal2004entanglement,tomamichel2013one}. As such, the third party would fail to corroborate the responses of the first two parties. Intuitively, one can reduce the gap between the classical and quantum winning probabilities by expanding the number of provers. This provides a different method to develop information-theoretic secure multi-party interactive ZKPs without any computational hardness assumption. The authors in \cite{alikhani2021experimental} exploited this property to develop a three-prover relativistic ZKP for graph $3$-coloring that is sound against entangled provers. 

This analysis can be extended, as shown in \cite{kempe2011entangled}, to symmetric games. The authors constructed an efficient classical strategy for a two-prover symmetric game (of which we can use to prove soundness) from the three-party entangled game strategy with a similar winning probability due to the monogamy of entanglement property and gentle measurement lemma. In their classical strategy construction, the provers perform consecutive measurements on the shared tripartite entangled states for all possible questions and record the answers before the game begins. These recorded answers are then used to win the two-party symmetric non-local game without using any shared entangled state.

We borrow the result from \cite{kempe2011entangled} and develop a knowledge extractor for the relativistic graph $3$-coloring protocol. Our extractor queries only two of the three provers and repeatedly rewinds until all possible questions are asked. Later, it uses these answers to construct a witness. Using a similar line of argument to that presented in \cite{kempe2011entangled}, we prove the following theorem for the knowledge error of the quantum proofs of knowledge.

\begin{theorem}[Informal]
\label{thm:rzkp3col}
If the 2-prover relativistic zero-knowledge proof proposed in \cite{crepeau2019practical} for graph ($G=(\mathcal{V},H)$) 3-coloring has a classical proof of knowledge with knowledge error $\kappa_c$, then there exist a 3-prover relativistic zero knowledge proof that has a quantum proof of knowledge with knowledge error $\kappa_q$, such that
\begin{equation}
\kappa_q = \kappa_c + \delta, 
\end{equation}
where $\delta \geq \frac{16}{2401|H|^4}$.
\end{theorem}

\subsubsection{Improved Soundness Bounds from a New Consecutive Measurement Lemma:}

When analyzing the quantum rewinding property of the relativistic $\Sigma$-protocols with $\F_Q$-commitment, we use a $t$-consecutive measurement theorem similar to \cite{don2019security} for $t=2$. 
We study the two-consecutive measurement case in more detail, with two sets of measurement operators $\{W_1^{s_1}\}_{s_1\in S_1}$ and $\{W_2^{s_2}\}_{s_2\in S_2}$, where $W_i^{s_i}$ corresponds to the $i$-th measurement with outcome $s_i$, which can be assumed to be projective without loss of generality.
For this special case, we derive a tighter lower bound than in \cite{don2019security}.  

\begin{theorem}[Informal]
\label{thm:consec_inf}
    Consider two sets of orthogonal projectors $\{W_1^{s_1}\}_{s_1\in S_1}$ and $\{W_2^{s_2}\}_{s_2\in S_2}$ with $W_i=\sum_{s_i=1}^{S_i}W_i^{s_i}$ and $W_i^sW_i^{s'}=\delta_{ss'}W_i^s$. 
    Let $F_1=\frac{1}{2}\left(\Tr[W_1\sigma]+\Tr[W_2\sigma]\right)$, $F_1>\frac{1}{2}$ and 
    \begin{equation*}
        F_2=\frac{1}{2}\sum_{s_1=1,s_2=1}^{\abs{S_1}\abs{S_2}}\{\Tr[W_2^{s_2}W_1^{s_1}\sigma W_1^{s_1}]+\Tr[W_1^{s_1}W_2^{s_2}\sigma W_2^{s_2}]\}.
    \end{equation*}
    Then, we have
    
    \begin{equation*}
        F_2\geq\frac{2}{\max_i \abs{S_i}}\left(F_1-\frac{1}{2}\right)^2.
    \end{equation*}

\end{theorem}

Note that, Unruh~\cite{unruh2012quantum} achieved a similar bound, for the cases where $\abs{S_1} = 1$, and $\abs{S_2} = 1$.

Our lower bound from \Cref{thm:consec_inf} can drastically improve the soundness error of the relativistic $\Sigma$-protocols for Hamiltonian Cycle~\cite{chailloux2017relativistic} and Subset Sum~\cite{crepeau2023zero}.

\begin{lemma}[Informal]
The $\F_Q$-commitment based two-prover relativistic zero-knowledge proof protocol for the Hamiltonian cycle that is proposed in~\cite{chailloux2017relativistic} has the following soundness error against the entangled provers. $$\frac{1}{2} + \left(\frac{n!}{2Q}\right)^{\frac{1}{2}},$$
where $n$ is the number of nodes in the underlying graph.
\end{lemma}

Note that, in~\cite{chailloux2017relativistic} the soundness error was proven to be $\frac{1}{2} + \left(\frac{64n!}{Q}\right)^{\frac{1}{3}}$. Similarly in~\cite{crepeau2023zero}, Crepeau et al. derived a soundness error $\frac{1}{2} + \left(\frac{64 \cdot 2^n}{Q}\right)^{\frac{1}{3}}$, where $n$ is the number of elements in the input set of the subset-sum problem, and $Q$ is the parameter of the relativistic $\F_Q$-bit commitment protocol. \Cref{thm:consec_inf} gives the following soundness error for the subset-sum problems. 

\begin{lemma}[Informal]
The $\F_Q$-commitment based two-prover relativistic zero-knowledge proof protocol for the subset-sum that is proposed in~\cite{crepeau2023zero} has the following soundness error against the entangled provers $$\frac{1}{2} + \left(\frac{2^{n-1}}{Q}\right)^{\frac{1}{2}},$$
where $n$ is the size of the set.
\end{lemma}

Notably, these improvements directly result in a reduction in the value of $Q$ necessary to obtain a soundness error negligibly greater than $\frac{1}{2}$, operationally corresponding to improvements in communication complexity. We summarize the comparison between our analysis and previous results in \Cref{tab:comp-results}.

 \begin{table}
     \centering
     \begin{tabular}{|c|c|c|c|c|}
     \hline
        Problem  & Prev. Soundness & Our Soundness  & Prev. $Q$ & Our $Q$  \\ \hline
        Hamiltonian Cycle & $\frac{1}{2} + \left(\frac{64n!}{Q}\right)^{\frac{1}{3}}$\cite{chailloux2017relativistic} & $\frac{1}{2} + \left(\frac{n!}{2Q}\right)^{\frac{1}{2}}$ & $2^{3\eta+6}n!$ & $2^{2\eta-1}n!$ \\\hline
    
        Subset-Sum & $\frac{1}{2} + \left(\frac{64 \cdot 2^n}{Q}\right)^{\frac{1}{3}}$ \cite{crepeau2023zero}  & $\frac{1}{2} + \left(\frac{2^{n-1}} {Q}\right)^{\frac{1}{2}}$ & $2^{n+3\eta+6}$ &  $2^{n+2\eta-1}$  \\\hline

     \end{tabular}
     \vspace{1.2mm}
     \caption{We compare the bound on soundness error obtained by our analysis to previous analyses of the same protocols. In addition, we compare the value of $Q$ necessary to achieve a soundness error of $\frac{1}{2}+2^{-\eta}$ for security parameter $\eta$ based on each bound.}
     \label{tab:comp-results}
 \end{table}
\section{Preliminaries}
\label{sec:prelims}
We now introduce a collection of definitions and results which we will make use of in later sections. 

\begin{definition} [Relativistic commitment scheme] \label{def:rel-com-scheme}
    A relativistic commitment scheme described by $(\textsf{Com}_K, \textsf{Open}_K)$ is an interactive protocol between two, potentially quantum, senders $\langle S_1,S_2\rangle$ who share randomness $K$, and receiver $R$ with two phases. The scheme has both a Commit phase and an Open phase.
    \begin{enumerate}
        \item (Commit phase.) $S_1$ chooses a bit $d\in \{0,1\}$ that she wants to commit to. $R$ sends some randomness $rand$ to $S_1$ and $S_1$ generates $com = \textsf{Com}_K(d,rand)$ back as a commitment. 
        \item (Open phase.) $S_2$ reveals the value and associated opening information of $d, open = \textsf{Open}_K(com)$ back to $R$. Then, depending on the commitment information previously sent by $S_1$ and the opening information sent by $S_2$, $R$ outputs either ``Accept" or ``Reject".
    \end{enumerate}
\end{definition}
We may omit explicitly mentioning the pre-shared randomness in our discussions for brevity, and describe an honest bit commitment scheme just as 
$(\textsf{Com},\textsf{Open})$. All commitment schemes we consider will be perfectly hiding and therefore we are only interested in the binding property. To prove the existence of QPoKs for all NP, we will require the notion of binding given below in~\Cref{def:fairly-binding}, which is equivalent to both the fairly-binding and simultaneous opening properties (due to the binary message space) given by Fehr and Fillinger \cite[Definition 3.2, Definition 3.1.3]{fehr2016composition}. 

\begin{definition} [$\epsilon_{FB}$-fairly-binding]\label{def:fairly-binding}
    A bit commitment scheme is $\epsilon_{FB}$-fairly-binding if for malicious senders $S^*_1,S^*_2$ with cheating strategy $(\textsf{Com}^*,\textsf{Open}^*)$,
    \begin{align*}
        \forall \textsf{Com}^*, & \max_{\textsf{Open}^*} \Pr_{resp}[S_1^*,S_2^* \text{ successfully reveals } d=0 \land d=1 | rand, \textsf{Com}^*, \textsf{Open}^*]\\
        &\leq \epsilon_{FB},
    \end{align*}
    where $resp = \textsf{Open}^*(\textsf{Com}^*(d))$ is the opening $S_2^*$ sends to the receiver $R$ for some $d$.
\end{definition}

This binding property roughly restricts simultaneous opening, i.e., ensuring that, for any $rand$ received by $S_1^*$ and any cheating strategy, senders will not be able to open a message to both $0$ and $1$ except with some probability $\epsilon_{FB}$. This fairly-binding definition is previously studied in~\cite{fehr2016composition}.

We remark that the widely used relativistic $\mathbb{F}_Q$ commitment scheme can be shown to be $\frac{1}{Q}$-fairly-binding, the proof of which follows from ~\cite[Proposition 3.15]{fehr2016composition}. For completeness, we give a description of the scheme below, noting that the commitment and open phases happen simultaneously:\\
\textbf{Relativistic $\mathbb{F}_Q$ commitment:}
\vspace{-2.5mm}
\begin{enumerate}
    \item (Preparation phase.) $S_1$ and $S_2$ agree on a committed bit $d \in \{0,1\}$, and share a random number $c \in \mathbb{F}_Q$.
    \item (Commitment phase.) $S_1$ recevies a random $a \in \mathbb{F}_Q$ from $R$. Then $S_1$ returns $w:= a \cdot d + c$ to $R$.
    \item (Open Phase.) $S_2$ sends $d,c$ to $R$. $R$ verifies that $w=a \cdot d + c$.
\end{enumerate}

We recall also the definition of quantum proofs of knowledge.
\begin{definition} [Quantum Proofs of Knowledge \cite{unruh2012quantum}] \label{def:QPoK}
    Let $L$ be a language in NP and let $\R$ be its corresponding witness relation (that is, $x \in L$ if and only if there exists a witness $w$ such that $(x, w) \in \R)$. We call a proof system $(P, V)$ a quantum proof of knowledge for $L$ with knowledge error $\kappa$ iff there exists a constant $d>0$, a polynomially-bounded function $p > 0$, and a quantum-polynomial-time oracle machine $K$ such that for any interactive quantum machine $P^*$, any polynomial $l$, any security parameter $\eta \in \N$, any state $|\psi\rangle$, and any $x \in \{0,1\}^{\leq l(\eta)}$, we have that
    
    \medskip\noindent \textbf{Non-triviality.} If $(x, w) \in \R$ then on public input $x$ and where  prover $P$ is given $w$ as private input, the verifier will accept with probability one.
    
    \medskip\noindent \textbf{Validity (with knowledge error $\kappa$).} For every, possibly malicious, prover $P^*$, and every $x$, the following holds
    \begin{align*}
    \Pr&\Big[\big\langle P^*(x, |\psi\rangle), V(x, r)\big\rangle =1\Big] > \kappa \Longrightarrow \\
    & \Pr[(x,w) \in \R: w \leftarrow K^{P^*(x,|\psi\rangle)}(x)] \geq \frac{1}{p(\eta)}\Big( \Pr\big[\left\langle P^*(x, |\psi\rangle), V(x, r)\right\rangle=1\big] - \kappa\Big)^d.
    \end{align*}
(where this probability is taken over the random choices $r$ of the verifier). The oracle machine $K$ is known as a knowledge extractor with running time polynomial in $\eta$ and the running time of $P^*$, and on input $x$ outputs a witness $w$ for $x$ (i.e. $w$ such that $(x, w) \in \R$). Here $P$ indicates generalized provers, which can be either one prover or multiple provers. The extractor $K$ can query all provers.
\end{definition} 

We will also make use of the below definition for $k$-player single round non-local games. 

\begin{definition}[Single Round $k$-Player Non-Local Games]
\label{def:symm_gaenene}
For some finite sets $I_Q, I_A$, a single round $k$-player non-local game $G^k=(I_Q,I_A,\pi,\V)$, played between $k$-spatially separated non-communicating players $P_1, \ldots, P_k$ and a verifier is given by a set of questions $I_Q$ and answers $I_A$, together with a distribution $\pi: [I_Q]^k \rightarrow [0,1]$ and a verification function $\V:[I_Q]^k \times [I_A]^k \rightarrow \{0,1\}$. It proceeds as follows, 
\begin{enumerate}
\item The verifier chooses $k$ inputs $(X_{1}, \ldots , X_{k}) \in [I_Q]^k$ from the distribution $\pi(X_{1}, \ldots , X_{k})$, and sends the $i$-th input $X_{i}$ to the $i$-th player $P_i$.
\item After receiving the questions, the $i$-th player sends its answer $O_i \in I_A$ to the verifier.
\item The verifier computes the function $\V(O_1,\ldots , O_k|X_1, \ldots , X_k)$.
\item The players win if and only if $\V(O_1,\ldots ,O_k|X_1, \ldots , X_k) = 1$.
\end{enumerate}
\end{definition}

Informally speaking, we refer to a $k$–prover interactive proof system for a language $L$ as symmetric if V can permute the questions to all provers without changing their distribution. This paper uses the definition of symmetric games from \cite{kempe2011entangled}.

\begin{definition}[Symmetric Games]
\label{def:symm}
A $k$-player single round non-local game $G^k=(I_Q,I_A,\pi,\V)$ is called symmetric, if the input distribution $\pi$ is a symmetric distribution. 
\end{definition}

In \cite{kempe2011entangled}, Kempe et al. propose a technique to calculate an upper bound on the winning probability of a $3$-player symmetric game. In this paper, we use the following Lemma from \cite{kempe2011entangled} to get a lower bound on the knowledge error of a symmetric $3$-prover relativistic zero-knowledge protocol.  

\begin{lemma}[\cite{kempe2011entangled}, Claim $20$]
\label{lem:meas_dist}
Suppose $|\Psi\rangle_{P_1P_2P_3}$ be a tri-partite state shared between three players $P_1, P_2,P_3$ such that $\rho_{P_1P_2} := \Tr_{P_3}\left[|\Psi\rangle\langle \Psi|_{P_1P_2P_3}\right]$. For some finite sets $I_Q, I_A$, we consider a set of operators $\{W_q\}_{q \in I_Q}$ such that for every $q \in I_Q$, $W_q :=\{W_q^o\}_{o\in I_A}$ is a set of projectors and $\sum_{o\in I_A} W_q^o = \id$. Then for any $q \in I_Q$,
\begin{equation}
\lVert\rho_{P_1P_2} - \sum_{o \in I_A}W^o_q \rho_{P_1P_2}W^o_q\rVert_1 \leq 6\sqrt{1 - \pi_2(q)},
\end{equation}
where $\pi_2(q) := \sum_{o\in i_a} \Tr[(W_q^o \otimes W_q^o) \rho_{P_1P_2}]$.
\end{lemma}

In the literature of both relativistic cryptography and quantum proofs of knowledge, there are some existing lower bounds on the average winning probability certain events being observed after consecutive measurements.
We detail a selection of those relevant to us here, starting with a lower bound given by Chailloux and Leverrier for any set of projectors \cite{chailloux2017relativistic}.

\begin{theorem}[\cite{chailloux2017relativistic}, Theorem $1$]
\label{thm:cha_thm1}
Consider a set of $c$ projectors $W_1, \ldots , W_c$ such that for every $i \in \{1, \ldots, c\}$ $W_i := \sum_{s=1}^{|S_i|} W_i^s$, where the projectors $\{W_i^s\}_{s \in S_i}$ are orthogonal projectors. Let $\sigma$ be any quantum state, and $F_1 := \frac{1}{c} \sum_{i \in C} \Tr[W_i \sigma]$, and $F_2 := \frac{1}{c(c-1)} \sum_{\substack{i,j \in C \\ i \neq j}}\sum_{s \in S_i,s' \in S_j}\Tr[W_jW_i \sigma W_i]$. If $F_1 \geq \frac{1}{c}$, then
\begin{equation}
F_2 \geq \frac{1}{64|S|}\left(F_1 - \frac{1}{c}\right)^3,
\end{equation}
where $|S| := \max\{|S_1|, \ldots , |S_c|\}$.
\end{theorem}

Later in \cite{don2019security}, Don et al. provides another lower bound for general $t$-consecutive measurements in a slightly different setting.

\begin{theorem}[\cite{don2019security}, Lemma $29$]
Consider a set of $c$ projectors $W_1, \ldots, W_c$, and let $|\Psi\rangle$ be any quantum state. Define:
$$
F_1 := \frac{1}{c} \sum_{i \in C} \Tr[W_i |\Psi\rangle\langle \Psi|] = \frac{1}{c} \sum_{i \in C} \lVert W_i |\Psi\rangle \rVert^2,
$$
and for any $t > 0$, define:
$$
F_t := \frac{1}{c^t} \sum_{i_1, \ldots, i_t} \lVert W_{i_t} \ldots W_{i_1} |\Psi\rangle \rVert^2.$$ Then we have
\begin{equation}
\label{thm:cha_consec}
F_t \geq F_1^{(2t - 1)}.
\end{equation}
\end{theorem}

\section{Quantum Proofs of Knowledge in the Relativistic Setting}
\label{sec:QPoKs}
We begin by formalizing the concept of interactive machines interacting under relativistic constraints, as well as quantum oracle machines with rewinding capabilities in the relativistic setting. With these definitions in hand, we proceed to prove that relativistic $\Sigma$-protocols which have $\delta_{SS}$-special soundness for appropriate $\delta_{SS}$ are indeed quantum proofs of knowledge, answering our motivating question in the affirmative for a large class of protocols, including those proposed in \cite{chailloux2017relativistic,crepeau2023zero}. 
Moreover, we give a proof that there exist relativistic QPoKs for all of NP assuming the existence of a relativistic bit commitment scheme with the fairly-binding property in \Cref{def:fairly-binding}. We arrive at this result via a construction from this binding property of such a proof system for Hamiltonian Cycle. We show that the proof system has special soundness and therefore, by our previous result, it is a QPoK. Finally, we consider the existence of QPoKs for protocols without special soundness, such as the one proposed in \cite{crepeau2019practical} for graph 3-coloring, and construct extractors for a class of protocols which are symmetric and has a classical proof of knowledge. 

\subsection{Interactive Machines and Rewinding Oracles with Relativistic Constraints}
\label{sec:QPoKs-Defs}

\subsubsection{Execution of three interactive machines in relativistic setting.} 

In order to analyze the operation of the relativisticly constrained machines interacting as part of \cref{def:rel-sigma-protocol}, we first formalize a model for the interaction of such machines. In fact, for simplicitly we define only configurations of three machines where one machine interacts with two other machines which do not interact with one another. We argue that this model is equivalent to the setting relevant to us, due to the relativistic checking in \Cref{def:rel-sigma-protocol}. We note that this definition can easily generalize to additional parties and arbitrary interactions between the parties, but the below is sufficient for our purposes.

We define the output of three spacelike separated machines $\mathsf{M_0}, \mathsf{M_1}$ and $\mathsf{M_2}$ as follows, where the action of each machine $M_i$ on round $j$ is exactly determined by circuits $M_{i\eta c,j}$ for security parameter $\eta$ and input $c$. We define their execution below
$$\langle \mathsf{M_0}(x,\ket{\Phi}),  \mathsf{M_1}(y,\ket{\Phi}),  \mathsf{M_2}(z,\ket{\Phi})\rangle$$ by the following process:
\begin{enumerate}
    \item Initialize quantum registers $S_0,S_1,S_2,N_1,N_2$ as $\ket{\Phi}_{S_0S_1S_2},\ket{0},\ket{0}$, respectively, where $\ket{\Phi}$ describes a potentially non-separable global system with subsystem $S_i$ local to $\mathsf{M}_i$. 
    \item For $j=1,\cdots,r$, sequentially apply: (1) $\mathsf{M}_{0\eta x,j}$ to $S_0$, $N_1$, and $N_2$ (2) $\mathsf{M}_{1\eta y,j} \otimes \mathsf{M}_{2\eta z,j}$ with each operator in the tensor $\mathsf{M}_{i\eta c,j}$ acting only on spaces $S_i$ and $N_i$.
    \item Measure $S_1,S_2,S_3$ in the computational basis. The random variable $$\langle \mathsf{M_0}(x,\ket{\Phi}), \mathsf{M_1}(y,\ket{\Phi}),  \mathsf{M_2}(z,\ket{\Phi})\rangle$$ denotes the classical outcome of the measurement. 

\end{enumerate}

We note that the operation in the second step of the above definition can be understood as $\mathsf{M}_1$ simultaneously and separately communicating with machines $\mathsf{M}_2$ and $\mathsf{M}_3$, as the first operator acts on both message registers. Similarly, the second operation in that step denotes the simulatenous response by both machines to machine $\mathsf{M}_1$. This definition of the actions of machines $\mathsf{M}_1$ and $\mathsf{M}_2$ implicitly enforces their non-communication within rounds, justified in the relativistic setting by the following logic.
During each round $j$, machines $\mathsf{M}_1$ and $\mathsf{M}_2$ are spacelike separated via an appropriate timing constraint while they apply the operations $\mathsf{M}_{1\eta x,j}$ and $\mathsf{M}_{2\eta x,j}$.
As such, these operations are restricted to the local subsystems they have access to, namely $S_i$ and $N_i$ for machine $\mathsf{M}_i$, which we can model as the application of $\mathsf{M}_{1\eta x,j}\otimes\mathsf{M}_{2\eta x,j}$.
We note that this is similar to the execution of some explicitly non-communicating $\mathsf{M}_1$ and $\mathsf{M}_2$ during each protocol step $j$.
Importantly, the inputs for the round, stored in classical registers $N_1$ and $N_2$, are utilised by the two circuits independently, indicating that one machine would not be privvy to the inputs of the other machine that it is not meant to know (i.e. that the machine $\mathsf{M}_0$ did not share with it) within the round.
We note that this guarantee would not hold between rounds, e.g. the response of $\mathsf{M}_1$ in the $j$-th round may be accessible to machine $\mathsf{M}_2$ at the $(j+1)$-th round, which reflects how this would occur in practice.
Since the time taken between the response of $\mathsf{M}_1$ to the arrival of the next round's message to machine $\mathsf{M}_2$ is sufficient for a signal to be transmitted from $\mathsf{M}_1$ to $\mathsf{M}_2$, we can make no relativistic guarantees here.
In fact, we note that in general the definition allows $\mathsf{M}_0$ to pass on communication from a party in one round to the other party in a subsequent round, which can be simulated by $\mathsf{M}_0$ copying part of the message from $N_1$ to $N_2$ and vice versa.

\subsubsection{Oracle algorithms with rewinding}
\label{sec:OraclewRewinding}

We describe a quantum oracle machine $\mathsf{A}$ with oracle access to multiple quantum interactive machines as a family of circuits, $(\mathsf{A}_{\eta x})$, operating on quantum registers $S_{\mathsf{A}}$, $N_1$, $N_2$, $S_{\mathsf{M_1}}$, $S_{\mathsf{M_2}}$. The circuit $\mathsf{A}_{\eta x}$ contains both normal gates operating only on $S_{\mathsf{A}}$ and $N_1,N_2$, and special gates $\circ_i$ and $\circ_i^{\dagger}$ which represents the activation of the oracle access to $M_i$ for $i \in \{1,2\}$. 
These special gates operate on one qubit of $S_{\mathsf{A}}$ and on the whole of $(N_i,S_{\mathsf{M_i}})$ for $i \in {1,2}$ correspondingly. The execution of $\mathsf{A}^{\mathsf{M_1}(x',\ket{\Phi}), \mathsf{M_2}(x',\ket{\Phi})}$, where $\ket{\Phi}$ is a joint state, is defined as follows:
\begin{enumerate}
    \item Initialize $S_{\mathsf{A}},N_1,N_2,S_{\mathsf{M_1}}, S_{\mathsf{M_2}}$ with $\ket{0},\ket{0},\ket{0}, \ket{\Phi}_{S_{\mathsf{M_1}}S_{\mathsf{M_2}}}$.
    \item Execute the circuit $\mathsf{A}_{\eta x}$. 
    \begin{itemize}
        \item When the gate $\circ_1$ is to be applied on $C,N_1,S_{\mathsf{M_1}}$ where $C$ is a (control) qubit in $S_{\mathsf{A}}$, apply the unitary transformation $U$ defined by
        \begin{align*}
        U(\ket{0}\otimes\ket{\psi}_{N_1}\otimes\ket{\phi}_{S_{\mathsf{M_1}}S_{\mathsf{M_2}}}) &:= \ket{0}\otimes\ket{\psi}_{N_1}\otimes\ket{\phi}_{S_{\mathsf{M_1}}S_{\mathsf{M_2}}}\\
        U(\ket{1}\otimes\ket{\psi}_{N_1}\otimes\ket{\phi}_{S_{\mathsf{M_1}}S_{\mathsf{M_2}}}) &:= \ket{1}\otimes \mathsf{M}_{1 \eta x',j}\left(\ket{\psi}_{N_1}\otimes\ket{\phi}_{S_{\mathsf{M}_1}S_{\mathsf{M_2}}}\right),
        \end{align*}
        where $\mathsf{M}_{1 \eta x',j}$ is only applied on registers $N_1$ and $S_{\mathsf{M}_1}$. Intuitively, this represents an application of the oracle given to $\mathsf{A}$ and it can choose whether to activate the oracle access. We note that if $N_1$ is a classical register, we can write $\mathsf{M}_{1\eta x',j}=\sum_{x}\dyad{x}_{N_1}\otimes \mathsf{M}^x_{1\eta x',j}$, where $\mathsf{M}^x_{1\eta x',j}$ is a unitary operating on $S_{\mathsf{M}_1}$.
        \item $\circ_1^{\dagger}$ is treated analogously, except that we use $\mathsf{M}_{1 \eta x',j}^{\dagger}$ instead of $\mathsf{M}_{1 \eta x',j}$. Intuitively, this corresponds to $\mathsf{A}$ rewinding an application of $\mathsf{M}_1$ based on the control bit $C$. 
        \item $\circ_2$ and $\circ_2^{\dagger}$ are also treated analogously, except that $\mathsf{M}_{2 \eta x',j}$ and $\mathsf{M}_{2 \eta x',j}^{\dagger}$ are applied on registers $N_2$ and $S_{\mathsf{M_2}}$.
    \end{itemize}
    \item Finally, we measure $S_{\mathsf{A}}$ in the computational basis, and the random variable $\mathsf{A}^{\mathsf{M_1}(x',\ket{\Phi}), \mathsf{M_2}(x',\ket{\Phi})}$ describes the outcome of the measurement. 
\end{enumerate}

We say that the algorithm $\mathsf{A}$ is a quantum-poly-time algorithm if given $\eta$ and $x$, it can be described in deterministic time $O(|x| + \eta)$ and the circuits $\mathsf{A}_{\eta x}$ have size bounded by $O(|x|+\eta)$.

\subsection{QPoKs for Relativistic $\Sigma$-Protocols with $\delta_{SS}$-Special Soundness}\label{sec:QPoKswithSpec}

For our extractor, we generalize the construction of a canonical extractor from a special extractor due to Unruh \cite{unruh2012quantum} to the two-prover, relativistic setting, noting that in this setting, our generalized construction works for $\Sigma$-protocols that have only probabilistic special soundnness (\cref{defn:dss-special-soundness}) and that it does not require strict soundness, as was required in the non-relativistic setting. At a high-level, the canonical extractor works by interfacing with the provers (over dedicated message registers) over multiple rounds. In the first round, the extractor exactly plays the role of the verifier, recording one transcript. The extractor than rewinds the second prover, and re-interacts with it to generate a second transcript. By the special soundness of the proof system, if both transcripts are accepting then the canonical extractor can use the special extractor to generate a witness from the two transcripts. We formally model the operation of the extractor below.  

\subsubsection{Canonical extractor for Relativistic $\Sigma$-Protocols:}

Let $\langle (P_1, P_2),V\rangle$ be a $\Sigma$-protocol with $\delta_{SS}$-special soundness. Let $K_0$ be the special extractor for that protocol. We use it to construct the canonical extractor $K$ for $\langle (P_1, P_2),V\rangle$. $K$ will use measurements for the sake of presentation. 

Given oracle access to any provers $( P^*_1,P^*_2)$, the extractor $K^{(P^*_1,P^*_2)(x,\ket{\Phi})}(x)$ operates in general on four quantum registers $N_1,N_2,S_{P^*_1},S_{P^*_2}$, as presented in \Cref{sec:OraclewRewinding}. 
We leave out $S_K$ for brevity since the extractor only needs to select the challenges.
$S_{P^*_1},S_{P^*_2}$ are internal quantum registers of the respective systems, and $N_1$ and $N_2$ are used for communication between the extractor and $P^*_1$ and $P^*_2$, respectively. 
For clarity, we shall later expand $N_1=N_1N_3$ and $N_2=N_2N_2'N_4N_4'$, where $N_1$ is the message from $K$ to $P^*_1$, and $N_3$ is the response, while $N_2$ and $N_2'$ are the two challenges from $K$ to $P^*_2$, with $N_4$ and $N_4'$ storing the respective responses.
We assume that these communication registers are classical (as only classical message are sent between prover and verifier in practice), and generate them only when we need them.
Registers $S_{P^*_1},S_{P^*_2}$ are initialized to a quantum state $\ket{\Phi}_{S_{P^*_1},S_{P^*_2}}$. Let $P^*_{i\eta x}$ denote the unitary transformation describing a single activation of $P^*_i$. We now show how $K$ works:
\begin{enumerate}
    \item The extractor $K$ randomly selects a query $rand\in\mathcal{R}$ and prepares it in register $N_1$. Then, it activates $P_1^*$, taking as input $N_1$, and operating on its subsystem $S_{P_1^*}$, i.e. applying $\tilde{P}^*_{1\eta x}=\sum_{rand\in\mathcal{R}}\dyad{rand}_{N_1}\otimes P^{rand*}_{1\eta x}$ to the joint state
       \begin{align*}
        &\tilde{P}^*_{1\eta x}\left(\frac{1}{\abs{\mathcal{R}}}\sum_{rand\in\mathcal{R}}\dyad{rand}_{N_1}\otimes\dyad{\Phi}_{S_{P^*_1},S_{P^*_2}}\right)\\
        &=\frac{1}{\abs{\mathcal{R}}}\sum_{rand\in\mathcal{R}}\dyad{rand}_{N_1}\otimes P^{rand*}_{1\eta x}\dyad{\Phi}_{S_{P^*_1},S_{P^*_2}}\left(P^{rand*}_{1\eta x}\right)^{\dagger}.
    \end{align*}
    
    \item Now $S_{P_1^*}$ is expected to contain the commitment, which $P_1^*$ would in general measure (perhaps measuring only a subsystem of $P_1^*$) to generate the commitment $com$ to send to the verifier. The commitment output is recorded in register $N_3$, taking the value $com$ with probability $p_{com}:= |\alpha_{com}|^2$.
    Let the post-measurement state shared between $P^*_1$, and $P^*_2$ corresponding to query $rand$ and measurement outcome $com$ be $\sigma_{S_{P^*_1},S_{P^*_2}}^{rand,com}$. Mathematically, we can represent the ensemble $\{(p_{com}, \sigma_{S_{P^*_1},S_{P^*_2}}^{com})\}_{com}$ using the following density operator
    \begin{equation*}
        \rho_2=\sum_{rand,com}\frac{p_{com}}{\abs{\mathcal{R}}}\dyad{rand,com}_{N_1N_3}\otimes\sigma_{S_{P_1^*}S_{P_2^*}}^{rand,com}.
    \end{equation*}
    Moreover, we can give a mathematical expression for $\sigma_{S_{P^*_1},S_{P^*_2}}^{com,rand}$,
    \begin{equation}
    \label{eq:sigma}
    \sigma_{S_{P_1^*}S_{P_2^*}}^{rand,com}=W_{com}P_{1\eta x}^{rand*}\dyad{\Phi}_{S_{P_1^{*}}S_{P_2^{*}}}P_{1\eta x}^{rand*\dagger}W_{com},
    \end{equation}
    where $W_{com}$ is a projector representing the measurement performed by $P_1^*$ to generate the $com$ commitment stored in $N_3$.
    
    \item The extractor $K$ chooses $ch,ch' \in C_{\eta x}$ uniformly at random with $ch\neq ch'$, where the size of the challenge space is $|C_{\eta x}| = c$, and stores the challenges in $N_2$ and $N_2'$. 
    The prover in general can receive the first challenge $ch$, and apply a general map $\tilde{P}^*_{2\eta x}=\sum_{ch}\dyad{ch}_{N_2}\otimes P^{ch*}_{2\eta x}$, where $P^{ch*}_{2\eta x}$ applies to the internal register $S_{P^*_2}$. We note that some part of $S_{P^*_2}$ contains the response. 
    Applying these maps, we can write the ensemble of states as
    \begin{align*}
        \rho_3=&\sum_{\substack{rand,com,\\ ch\neq ch'}}\frac{p_{com}}{\abs{\mathcal{R}}c(c-1)}\ketbra{rand,com,ch,ch'}{\cdot}_{N_1N_3N_2N_2'}\otimes P^{ch*}_{2\eta x}\sigma_{S_{P_1^*}S_{P_2^*}}^{rand,com}P^{ch*\dagger}_{2\eta x},
    \end{align*}
    where we label $\dyad{x}_{X}=\ketbra{x}{\cdot}_{X}$ for brevity.
    \item Now $S_{P_2^*}$ is expected to contain the response $resp$ to $ch$. $P_2^*$ then measures $S_{P_2^*}$ (or perhaps a subsystem in $S_{P_2^*}$) in the computational basis, and denote the result as $resp$, and store the measurement result in a new register $N_4$, i.e. measurement $\mathcal{M}(\rho)=\sum_{resp}\dyad{resp}_{N_4}\otimes W_{resp}\rho W_{resp}$, where $W_{resp}$ is a projector applied on $S_{P_2^*}$ (or a subsystem of it). The post-measurement state is then 
     \begin{align*}
        &\rho_4 =\sum_{\substack{ch\neq ch',\\rand,com\\resp}}\hspace{-0.1in}\frac{p_{com}}{\abs{\mathcal{R}}c(c-1)}\ketbra{rand,com,ch,ch',resp}{\cdot}_{N_1N_3N_2N_2'N_4}\otimes W_{resp}P^{ch*}_{2\eta x}\sigma_{S_{P_1^*}S_{P_2^*}}^{rand,com}P^{ch*\dagger}_{2\eta x}W_{resp}
    \end{align*}
    
    \item The extractor $K$ applies $\tilde{P}^{*\dagger}_{2\eta x}$ for the rewinding of $P_2^*$. After the rewinding, we get the quantum state
    \begin{align*}
        \rho_5&=\tilde{P}^{*\dagger}_{2\eta x} \rho_4\tilde{P}^{*}_{2\eta x}=\hspace{-0.1in} \sum_{\substack{ch\neq ch'\\ rand,com\\resp}}\hspace{-0.1in} \frac{p_{com}}{\abs{\mathcal{R}}c(c-1)}\ketbra{rand,com,ch,ch',resp}{\cdot}_{N_1N_3N_2N_2'N_4}\\
        &\otimes P^{ch*\dagger}_{2\eta x}W_{resp}P^{ch*}_{2\eta x}\sigma_{S_{P_1^*}S_{P_2^*}}^{rand,com}P^{ch*\dagger}_{2\eta x}W_{resp}P^{ch*}_{2\eta x}.
    \end{align*}
    
    \item Finally, the extractor $K$ triggers the malicious prover $P_2^*$ to apply $\tilde{P}^*_{2\eta x,p}=\sum_{ch}\dyad{ch}_{N_2'}\otimes P^{ch*}_{2\eta x}$ to $N_2'$, $S_{P^*_2}$, with the response being part of $S_{P^*_2}$.
    The subsystem $S_{P^*_2}$ is then measured and the response is stored in register $N_4'$. 
    
    We can write the ensemble as 
    \begin{align*}
        \rho_6&=\hspace{-0.1in} \sum_{\substack{ch\neq ch'\\rand,com\\ resp,resp'}}\hspace{-0.1in}\frac{p_{com}}{\abs{\mathcal{R}}c(c-1)}\ketbra{rand,com,ch,ch',resp,resp'}{\cdot}_{N_1N_3N_2N_2'N_4N_4'}\\
        &\otimes W_{resp'}P^{ch'*}_{2\eta x}P^{ch*\dagger}_{2\eta x}W_{resp}P^{ch*}_{2\eta x}\sigma_{S_{P_1^*}S_{P_2^*}}^{rand,com}P^{ch*\dagger}_{2\eta x}W_{resp}P^{ch*}_{2\eta x}P^{ch'*\dagger}_{2\eta x}W_{resp'},
    \end{align*}
    
    \item $K$ outputs $w := K_0(x, rand, com, ch, resp, ch', resp')$, where the classical values of the inputs can be obtained from classical registers $N_1N_3N_2N_2'N_4N_4'$.
\end{enumerate}

\subsubsection{Lower Bound on the Knowledge Error}
We now proceed to prove the following bound on the knowledge error of the extractor extractor defined above, establishing the existence of QPoKs for the considered class of protocols. 
\begin{theorem} \label{thm:sig-QPoK}
    Let $\langle (P_1,P_2),V \rangle$ be a relativistic $\Sigma$-protocol for a relation $R$ with challenge space $C_{\eta x}$. Assume that $\langle (P_1,P_2),V \rangle$ has $\delta_{SS}$-special soundness with $\delta_{SS} < \frac{1}{p(\eta)},$ where $p(\eta)$ represents some bounded polynomial of $\eta$. Then, $\langle (P_1,P_2),V \rangle$ has a quantum proof of knowledge with knowledge error $\frac{1}{c}$ where $c = |C_{\eta x}|$.
\end{theorem}

\begin{proof}
We first note that since $\langle P,V\rangle$ has $\delta_{SS}$-special soundness, if both responses with $ch \neq ch'$ are acceptable, the extractor $K$ will extract a valid witness with probability at least $1-\delta_{SS}$. 
We also note that to prove knowledge error of $\frac{1}{c}$, we begin with the assumption that $\Pr[V\,accepts]>\frac{1}{c}$.
Define $$\textsf{RA}^{rand,ch}_{com} := \{resp: (rand,com,ch,resp)\text{ is an acceptable instance}\},$$ 
then we have 
\begin{align*}
    &\Pr_{rand, com}[K^{\langle P^*_1,P^*_2\rangle(x,\ket{\Phi})}(x) \text{ outputs a valid witness}] \\
    \geq &\sum_{\substack{rand,com\\ ch,ch'}}\Tr[\left(\ketbra{rand,com,ch,ch'}{\cdot}_{N_1N_3N_2'N_2'}\otimes\Pi_{N_4N_4'}^{\substack{resp\in \textsf{RA}^{rand,ch}_{com}\\resp'\in \textsf{RA}^{rand,ch'}_{com}}}\right)\rho_6]\times\left( 1-\delta_{SS}\right)\\
    =&\left( 1-\delta_{SS}\right) \hspace{-0.1in}\sum_{rand,com}\hspace{-0.05in}\frac{p_{com}}{\abs{\mathcal{R}}}\left\{\frac{1}{c(c-1)}\hspace{-0.3in}\sum_{\substack{ch,ch',\\ ch\neq ch'\\ resp\in\textsf{RA}^{rand,ch}_{com}\\ resp'\in\textsf{RA}^{rand,ch'}_{com}}}\hspace{-0.3in}\Tr[W^{'ch'}_{resp'}W^{'ch}_{resp}\sigma_{S_{P_1^*}S_{P_2^*}}^{rand,com}W^{'ch}_{resp}W^{'ch'}_{resp'}]\right\}
\end{align*} 
where we check the recorded $resp$ and $resp'$ outcome stored in $N_4$ and $N_4'$ and $\Pi_{N_4N_4'}^{\substack{resp\in \textsf{RA}^{rand,ch}_{com}\\resp'\in \textsf{RA}^{rand,ch'}_{com}}}$ is a projector onto cases where $resp$ and $resp'$ belong to their corresponding accepting sets.
We also define $W^{'ch}_{resp} = P^{ch*\dagger}_{2\eta x}W_{resp}P^{ch*}_{2\eta x}$. 
Since $W_{resp}$ is an orthogonal projector and $ P^{ch*}_{2\eta x}$ is an unitary, $W^{'ch}_{resp}$ is also an orthogonal projector. 

We now consider a honest verifier $V$. The probability that it accepts the response from provers is
\begin{align*}
    &\Pr[V \text{ accepts}] \\
    &=\sum_{rand,com} \frac{p_{com}}{\abs{\mathcal{R}}} \left(\frac{1}{c} \hspace{-0.1in}\sum_{\substack{ch\\ resp \in  \textsf{RA}^{rand,ch}_{com}}} \hspace{-0.3in}\Tr[W_{resp}(P^{ch*}_{2\eta x})\sigma^{rand,com}_{S_{P_1^*}S_{P_2^*}} P^{ch*\dagger}_{2\eta x}W_{resp}]\right)\\
    &=\sum_{rand,com} \frac{p_{com}}{\abs{\mathcal{R}}} \left(\frac{1}{c} \hspace{-0.1in}\sum_{\substack{ch\\ resp \in  \textsf{RA}^{rand,ch}_{com}}} \hspace{-0.3in}\Tr[W_{resp}^{'ch}\sigma^{rand,com}_{S_{P_1^*}S_{P_2^*}}]\right)
\end{align*}     
By \Cref{thm:cha_thm1}, 
\begin{align*}
    &\Pr[K^{(P^*_1,P^*_2)(x,\ket{\Phi})}(x) \text{ outputs a valid witness}] \\
    \geq& \left( 1-\delta_{SS}\right) \sum_{\substack{rand\\ com}}\frac{p_{com}}{\abs{\mathcal{R}}}\left\{\frac{1}{64\abs{\textsf{RA}^{rand,ch}_{com}}}\left[\left(\frac{1}{c} \hspace{-0.1in}\sum_{\substack{ch\\ resp \in  \textsf{RA}^{rand,ch}_{com}}} \hspace{-0.3in}\Tr[W_{resp}^{'ch}\sigma^{rand,com}_{S_{P_1^*}S_{P_2^*}}]\right)-\frac{1}{c}\right]^3\right\}\\
    \geq& \left( 1-\delta_{SS}\right) \sum_{\substack{rand\\ com}}\frac{p_{com}}{\abs{\mathcal{R}}}\left\{\frac{1}{64\abs{\textsf{RA}_{max}}}\left[\left(\frac{1}{c} \hspace{-0.1in}\sum_{\substack{ch\\ resp \in  \textsf{RA}^{rand,ch}_{com}}} \hspace{-0.3in}\Tr[W_{resp}^{'ch}\sigma^{rand,com}_{S_{P_1^*}S_{P_2^*}}]\right)-\frac{1}{c}\right]^3\right\}\\
    \geq& \frac{1-\delta_{SS}}{64\abs{\textsf{RA}_{max}}}\left(\Pr[V\,accepts]-\frac{1}{c}\right)^3
\end{align*} 
where $\abs{\textsf{RA}_{max}}=\max_{rand,com,ch}\abs{\textsf{RA}^{rand,ch}_{resp}}$, and the final inequality comes from Jensen's inequality \cite{jensen1906fonctions}, the fact that $f(x)=(x-a)^3$ is convex on $x\in[a,\infty)$ and that $\Pr[V\,accepts]> \frac{1}{c}$. 
It naturally follows from \Cref{def:QPoK} that $\langle P,V \rangle$ has a quantum proof of knowledge with knowledge error $\frac{1}{c}$, concluding the proof. \qed
\end{proof}

\subsection{Relativistic QPoKs for NP from $\epsilon_{FB}$-Fairly-Binding Commitment Schemes}
\label{sec:QPoKs-NP}
We extend our results now by giving a construction of a $\Sigma$-protocol for Hamiltonian cycle which achieves special soundness with a negligible $\delta_{SS}$ using a general class of restricted relativistic commitment schemes, i.e. the $\epsilon_{FB}$-fairly-binding bit commitment schemes defined in \Cref{def:fairly-binding}. This protocol is therefore also a QPoK, but from a potentially more general class of commitments. As a corollary, we obtain QPoKs from such commitment schemes for all of NP. In \cite{fehr2016composition}, the authors show that the $\mathbb{F}_Q$-commitment schemes are also $\epsilon_{FB}$-fairly-binding. Therefore, all the existing relativistic sigma protocols that are designed using relativistic $\mathbb{F}_Q$-commitment scheme are quantum proofs of knowledge with knowledge error $\frac{1}{2}$.
 
\begin{protocol}{Relativistic Proof System for Hamiltonian Cycle \label{protocol:general-hc}}
\vspace{0.2mm}
$P_1$ and $P_2$ pre-agree on a Hamiltonian cycle $C$ of a given graph $G=(\mathcal{V},H)$, where  $\mathcal{V}$ is the vertex set, $H$ is the edge set, and a random permutation $\Pi$. Define $M$ to be an adjacency matrix of $\Pi(G)$, and $C':= \{ \Pi(i), \Pi(j): (i,j) \in C\}$. Denote the strategy that $P_1,P_2$ pre-agree on as  $(\textsf{Com},\textsf{Open})$, which utilizes an $\epsilon_{FB}$-fairly-binding commitment scheme for $\epsilon_{FB} < \frac{1}{p(\eta)}$ for some security parameter $\eta$.
\begin{enumerate}
    \item \text{[}Commitment to each bit of $M_{\Pi(G)}$.\text{]} $P_1$ first receives $rand$ from $V$. Then, $P_1$ uses $\textsf{Com}$ to commit to $\Pi$ and each entry of $M$ and sends the resulting commitments $com$ to $V$.
    \item \text{[}Challenge phase.\text{]} $V$ sends a uniform random bit $ch  \in\{0,1\}$ to $P_2$.
    \begin{itemize}
        \item If $ch=0, P_2$ sends $resp$ to $V$, which opens the commitments to $\Pi$ and $M$.
        \item If $ch=1, P_2$ sends $resp'$ to $V$, which only opens the commitments to $C'$ and to all $M_{i,j}$ with $(i,j) \in C'$.
    \end{itemize}
    \item \text{[}Check phase.\text{]} $V$ checks that those decommitments are valid and correspond to what $P_1$ have declared. $V$ also checks that the timing constraint of the bit commitment is satisfied. This means that
    \begin{itemize}
        \item If $ch=0$,  $V$ checks that the commitments $com$ are opened correctly, i.e. $\Pi$ is a permutation, and that $M$ is the adjacency matrix of $\Pi(G)$. 
        \item If $ch=1$, $V$ also checks that the commitments $com$ are opened correctly but in a different way, i.e. $C'$ is a Hamiltonian cycle, that exactly $M_{i,j}$ with $(i,j) \in C'$ are opened correctly and that $M_{i,j}=1 \ \forall (i,j) \in C'$.
    \end{itemize}
    If all checks passed, $V$ outputs $1$.
\end{enumerate}
\end{protocol}

\begin{Proposition}\label{prop:SpecSound-HC}
    Consider Protocol \ref{protocol:general-hc} with fairly-binding parameter $\epsilon_{FB}<\frac{1}{p(\eta)}$ and assume it has winning probability $\Pr[V\,accepts]>\frac{1}{c}+\xi$, where $\xi$ is any negligible quantity in $\eta$.
    The protocol is a relativistic $\Sigma$-Protocol with special soundness $\delta_{SS} <\frac{1}{p(\eta)}$. 
\end{Proposition}

\begin{proof}
Fix any value $rand$, commitment $com$. Suppose $(rand,com,ch,resp)$ and $(rand,com,ch',resp')$ are two accepting conversations for a given graph $G$ with $ch=0$ and $ch'=1$. For $ch=0$, the response $resp$ is a permutation $\Pi$, the corresponding adjacency matrix $M_{i,j}$ along with opening information $A$  which can be used to check the correctness of the commitments. For $ch'=1$, the response $resp'$ is a cycle $C'$ on the permuted graph $\Pi(G)$ and corresponding $A'$ to check that the commitments to each $M_{u,v}=1$ for all $(u,v) \in C'$ are opened correctly. 

We first note that $C'$ is a Hamiltonian cycle of $\Pi(G)$ iff $\Pi^{-1}(C')$ is a Hamiltonian cycle of the graph $G$ since $\Pi$ is a permutation matrix. As such, let us consider an extractor that extracts witness $$w :=K_{0}(G,rand,com,ch,resp,ch',resp') := \Pi^{-1}(C').$$The probability that the extractor is successful is 
\begin{align*}
    \Pr[(x,w)\in R|rand,com,acc]=&\Pr[C'\in HC(\Pi(G))|rand,com,acc]\\
    =&1-\Pr[C'\not\in HC(\Pi(G))|rand,com,acc],
\end{align*}
where $acc$ is the event that the opened values satisfy the proper relations with respect to the input graph, irrespective of if the commitments are opened successfully or not, $HC(\Pi(G))$ represents the set of Hamiltonian cycles in graph $\Pi(G)$, and the probability is evaluated over the choice of $ch$, $ch'$, $resp$ and $resp'$ for accepting conversations. If the cycle $C'$ provided is not a Hamiltonian cycle in $\Pi(G)$, then there exists at least an edge $(u,v)\in C'$ such that $\left(M^0_{\Pi(G)}\right)_{u,v}=0$ (for $ch=0$) while opened as $\left(M^1_{\Pi(G)}\right)_{u,v}=1$ (for $ch=1$), i.e. 
\begin{align*}
&\Pr_{resp,resp'}[(x,w)\in R|rand,com,acc] \\
& \geq 1-\Pr_{resp,resp'}\left[\left(\left(M^0_{\Pi(G)}\right)_{u,v}=0\right),\left(\left(M^1_{\Pi(G)}\right)_{u,v}=1\right)\Bigg|rand,com,acc\right],
\end{align*}
where by $\left(\left(M^{ch}_{\Pi(G)}\right)_{u,v}=d\right)$ we denote the event that, upon recieving challenge $ch$, the prover produces a valid opening to the value $d$ for the commitment to the edge $(u,v)$ in $M_{\Pi(G)}$, conditioned on the existing commitments $com$.
We note that the challenge and opening value directly coincide.

We now bound the probability
\begin{align*}
    &\Pr_{resp,resp'}\left[\left(\left(M^0_{\Pi(G)}\right)_{u,v}=0\right),\left(\left(M^1_{\Pi(G)}\right)_{u,v}=1\right)\Bigg|rand,com,acc\right] \\
    & \leq \frac{1}{p_{acc}}\Pr_{resp,resp'}\left[\left(\left(M^0_{\Pi(G)}\right)_{u,v}=0\right),\left(\left(M^1_{\Pi(G)}\right)_{u,v}=1\right),acc\Bigg|rand,com\right]\\
    & \leq \frac{1}{p_{acc}}\Pr_{resp,resp'}\left[\text{Successfully reveals}\,\left(\left(M^0_{\Pi(G)}\right)_{u,v}=0\right)\land\left(\left(M^1_{\Pi(G)}\right)_{u,v}=1\right)\Bigg|rand,com\right],
\end{align*}
where $p_{acc}$ is the probability of both transcripts being accepting.
In the second inequality, we focus on the acceptance conditions that checks for the successful reveal of $(M_{\Pi(G)}^0)_{u,v}$ and $(M_{\Pi(G)}^1)_{u,v}$ while ignoring the remaining checks.

We are therefore interested in the probability that $\left(M^d_{\Pi(G)}\right)_{u,v}=d$ is successfully revealed for both $d=0$ and $d=1$.
By the assumption that our underlying commitment scheme satisfies the $\epsilon_{FB}$-fairly-binding property of \Cref{def:fairly-binding}, for any commitment strategy used by the senders to generate the committed values, the maximal probability with which they can successfully convince a receiver, in this case the verifier, to accept two distinct openings for the same commitment is bounded from above by $\epsilon_{FB}$. We note that the distributions on $resp$ from the proof system implicitly include the relevant response distributions of the underlying bit commitment for edge $(u,v)$ as a marginal. As such, we bound
\begin{equation*}
\Pr_{resp,resp'}[(x,w)\in R|rand,com,acc] \geq 1 - \frac{\epsilon_{FB}}{p_{acc}},
\end{equation*}
which implies a special soundness of $\delta_{SS}=\frac{\epsilon_{FB}}{p_{acc}}$.

Using \Cref{thm:cha_thm1}, we can bound the probability of accepting both transcripts with the single-round winning probability,
\begin{equation*}
    p_{acc}\geq \frac{1}{64\abs{\textsf{RA}_{max}}}\left(\Pr[V\,accepts]-\frac{1}{c}\right)^3.
\end{equation*}
With the assumption that the winning probability satisfy $\Pr[V\,accepts]>\frac{1}{c}+\xi$, where $\xi$ is a negligible quantity (i.e. the difference is non-neglible) and given $\epsilon_{FB}<\frac{1}{p(\eta)}$ (i.e. $\epsilon_{FB}$ is negligible in $\eta$), we have that $$\delta_{SS}\leq64\abs{\textsf{RA}_{max}}\xi^{-3}\epsilon_{FB}$$ is negligible, noting that $\abs{\textsf{RA}_{max}}$ is independent of $\eta$~\cite{chailloux2017relativistic}.
\qed
\end{proof}

With a bound on the special soundness of Protocol \ref{protocol:general-hc}, we may obtain the following corollary.
\begin{corollary} [Relativistic QPoK for Hamiltonian Cycle]\label{cor:QPoK_HC}
    Let $(G,w) \in \text{Hamiltonian-Cycle}$ iff $w$ is a Hamiltonian cycle of graph $G$. Assuming the existence of a $\epsilon_{FB}$-fairly-binding bit commitment scheme with $\epsilon_{FB} < \frac{1}{p(\eta)}$, then there exists a QPoK for Hamiltonian-Cycle with knowledge error $\frac{1}{2}+\xi$, where $\xi$ is some negligible quantity in $\eta$.
\end{corollary}
\begin{proof}
    Let us consider the proof of knowledge with knowledge error $\frac{1}{c}+\xi$, where $\xi$ is negligible. 
    Since $\Pr[V\,accepts]>\frac{1}{c}+\xi$ and $\epsilon_{FB}<\frac{1}{p(\eta)}$ by assumption, from \Cref{prop:SpecSound-HC}, Protocol \ref{protocol:general-hc} is a $\Sigma$-Protocol with special soundness $\delta_{SS}<\frac{1}{p(\eta)}$. As such, we can apply \Cref{thm:sig-QPoK} to arrive at the result that Protocol \ref{protocol:general-hc} is a quantum proof of knowledge for Hamiltonian-Cycle with knowledge error $\frac{1}{c}+\xi$. \qed
\end{proof}
We can further extend the result to show the corollary below.

\begin{corollary} [Relativistic QPoK for all languages in NP]
    Let $R$ be an NP-relation. Assuming the existence of $\epsilon_{FB}$-fairly-binding bit commitment scheme with $\epsilon_{FB} < \frac{1}{p(\eta)}$, there is a zero-knowledge relativistic QPoK for $R$ with knowledge error $\frac{1}{2}+\xi$ for some negligible $\xi$.
\end{corollary}
\begin{proof}
    Since the Hamiltonian cycle problem is NP-complete, then $R$ can reduced to Hamiltonian Cycle problem in polynomial time. Then, combined with \Cref{cor:QPoK_HC}, it follows that there is a zero-knowledge QPoK for $R$ with knowledge error $\frac{1}{2} + \xi$ for some negligible $\xi$. \qed
\end{proof}

In \Cref{app:ex_sigma}, we prove that some existing protocols~\cite{chailloux2017relativistic,crepeau2023zero} for Hamiltonian Cycle and Subset-Sum which use $\mathbb{F}_Q$ commitments and have special soundness are QPoK, answering our motivating question regarding the validity of existing relativistic proof systems as QPoKs in the affirmative for these protocols.

\subsection{QPoKs for a Symmetric Relativistic ZKP Protocols without Special Soundness}
\label{sec:QPOKs-Symm}

We now consider relativistic quantum proofs of knowledge for protocols which do not have the special soundness property, and, as such, may require multiple rewindings to extract a witness. We conduct this analysis for a $3$-party relativistic zero knowledge proof system for graph 3-coloring that is proposed in \cite{crepeau2019practical}. 

To bound the knowledge error, we establish a relationship between the quantum proofs of knowledge of a $3$-prover protocol with the classical proof of knowledge of the $2$-prover relativistic zero knowledge protocol proposed in \cite{alikhani2021experimental}. More precisely, we show that if the $2$-prover relativistic ZKP protocol in \cite{alikhani2021experimental} is a classical PoK then the $3$-prover protocol in \cite{crepeau2019practical} is a quantum PoK.

We will first start with by describing the 2-prover relativistic ZKP protocol (Protocol 2) that was proposed in \cite{alikhani2021experimental,crepeau2023zero}. In Protocol 3, we then construct an extractor for that 2-prover protocol. Later we show that this is a classical PoK through a straightforward analysis. We then move to consider the 3-prover protocol (Protocol 4) that is proposed in \cite{alikhani2021experimental,crepeau2023zero}, which is a clear extension of Protocol 2. In Protocol 5 we construct an extractor for this 3-prover protocol using Protocol 3 as a subroutine. The extractor in Protocol 5 is inspired by the methods of \cite{kempe2011entangled}, which allow us to reduce the power of the entangled prover. We note that the application of that method here is non-trivial as its use in this setting is novel, and we defer the proof showing its applicability and quantifying the resultant gap in the classical and quantum knowledge errors to \Cref{sec:app-soundess}.

We reproduce the $2$-prover protocol in Protocol \ref{prot:3_prot_party} and the $3$-prover protocol in Protocol \ref{prot:3col} for convenience. We begin by considering the classical PoK of the $2$-prover protocol before turning to the quantum setting. 

\subsubsection{Classical Proof of Knowledge for the 3-Coloring Problem}

Let $\langle (P_1, P_2),V\rangle$ be the protocol as described in \cite{crepeau2019practical,alikhani2021experimental}.

\begin{protocol}{Interactive protocol for 3-COL: Two-prover\label{prot:3_prot_party}}
$P_1$ and $P_2$ pre-agree on a $3$-coloring $C_G$ of a graph $G = (\Vt,H)$, where $C_G := \{(i,c_i) |c_i \in \mathbb{F}_3\}_{i \in \Vt}$ such that if $(i,j) \in H$, then $c_i \neq c_j$. Two provers pre-share randomly selected labelings $l_i^0$ and $l_i^1$ for each $i\in \Vt$ such that $l_i^0 + l_i^1 = c_i (\textsf{mod } 3)$ holds.
\begin{enumerate}
  \item \textbf{Commit phase.}
  \begin{enumerate}
    \item
    $V$ picks $((i,j),b), ((i',j'),b')\in (H \times \mathbb{F}_2)^2$ according to distribution $\mathcal{D}_G$.
    \item
    $V$ sends $((i,j),b)$ to $P_1$ and $V$ sends $((i',j'),b')$ to $P_2$.
    \item 
    If $(i,j) \in H$ and $b \in  \mathbb{F}_2$ then $P_1$ replies $a_i=l_i^b$, $a_j=l_j^b$.
    \item 
    If $(i',j') \in H$ and $b' \in  \mathbb{F}_2$ then $P_2$ replies $a'_{i'}=l_{i'}^{b'}$, $a'_{j'}=l_{j'}^{b'}$.
  \end{enumerate}
  \item \textbf{Check Phase}
    \begin{enumerate}
    \item (Edge-Verification Test) If $(i,j)=(i',j')$ and $b\neq b'$, then $V$ accepts iff $a_i + a'_i \neq a_j + a'_j$. 
    \item (Well-Definition Test) 
    \begin{itemize}
        \item If $(i,j) = (i',j')$ and $b=b'$, then $V$ accepts iff $(a_i = a'_i) \wedge (a_j = a'_j)$.
        \item If $(i,j) \cap (i',j') = i$ and $b=b'$, then $V$ accepts iff $(a_i = a'_i)$.
        \item If $(i,j) \cap (i',j') =j$ and $b=b'$, then $V$ accepts iff $(a_j = a'_j)$.
    \end{itemize}
  \end{enumerate}
\end{enumerate}
\end{protocol}

We take the distribution of questions $\mathcal{D}_G$ similarly as in~\cite{grover2021further}. Then for the well definition test in which $e=(i, j) \in H$ and $e' \in \operatorname{Edges}(e)=\operatorname{Edges}(i) \cup \operatorname{Edges}(j)$, and for some parameter $\epsilon \in [0,1]$ we have:
\begin{equation}
\label{eq:wdt_dist}
\mathcal{D}_G\left(e, b, e^{\prime}, b\right)=\frac{1-\epsilon}{4|H|}\left(\frac{\left|\left\{e^{\prime}\right\} \cap \operatorname{Edges}(i)\right|}{|\operatorname{Edges}(i)|}+\frac{\left|\left\{e^{\prime}\right\} \cap \operatorname{Edges}(j)\right|}{|\operatorname{Edges}(j)|}\right),
\end{equation}
where $\operatorname{Edges}(v)$ refers to the set of neighbouring vertices, i.e. $\{i:(i,v)\in H\}$, and the factor of $1/4$ arises from choosing $b \in \{0,1\}$ and selecting $e'$ from the neighbors of either $i$ or $j$, each with a probability of $\frac{1}{2}$.

For the edge verification test, we have:
\begin{equation}
\label{eq:evt_dist}
\mathcal{D}_G(e, b, e, \bar{b})=\frac{\epsilon}{2|H|}+\frac{1-\epsilon}{4|H|}\left(\frac{1}{|\operatorname{Edges}(i)|}+\frac{1}{|\operatorname{Edges}(j)|}\right) \geq \frac{\epsilon}{2|H|}.
\end{equation}
We take the probability $\epsilon$ that $V$ chooses the edge verification test to be $\frac{1}{3}$.

We now construct an extractor $K$ which acts similarly to the verifier but rewinds to ask all possible questions.
    \begin{protocol}{Knowledge Extractor $K^{P^*}$ for Classical Provers in 3-Coloring}
    Input: Provers $P^*_1, P^*_2$ and the verifier receives $G=(\Vt,H)$
      \begin{enumerate}
        \item $P^*_1$ and $P^*_2$ are separated.
        \item $K$ interacts with $P^*_1$ and $P^*_2$ and queries $\left(((i, j), b),\left(\left(i^{\prime}, j^{\prime}\right), b^{\prime}\right)\right)$ chosen randomly from $ \mathcal{D}_G\left(H \times \mathbb{F}_2\right.$).
        \item $K$ receives $(a_i,a_j)$ and $(a'_{i'},a'_{j'})$ from $P^*_1$ and $P^*_2$ separately and performs either Edge Verification Test (go to step $4$) and Well-Definition Test (go to step $5$).
        \item If Edge Verification Test passes and the corresponding vertices are not marked, $K$ marks the corresponding vertices. If the test passes while the vertices are already marked, check whether conflicts exist. Suppose the coloring assignment for the marked vertex is $c_i$ ($c_j$). For such vertices there is a conflict if $a_i + a'_i \neq c_i$ ($a_j + a'_j \neq c_j$). If not, continue to step $6$. Otherwise, record both answers and leave a question mark on both vertices.
        \item If the Well-Definition Test passes, continue. Otherwise, leave a question mark on the vertex.
        \item Rewind to step $2$ to ask more questions.
        \item When all the questions are asked, check whether there exists any question mark on the vertex. If yes, $K$ marks the vertex with the most reasonable assignments based on recorded information. Then $K$ output vertices with corresponding marks.
      \end{enumerate}
    \end{protocol}
Suppose $\Pr\limits_{f_1,f_2}\left[\left\langle (P^*_1(G,f_1),P^*_2(G,f_2)), V\right\rangle=1\right]$ denotes the probability that classical provers (possibly malicious) convince the verifier, and $f_1,f_2$ denotes the strategy for $P^*_1$ and $P^*_2$. In Lemma \ref{lma:kappa}, we provide a upper bound on $$\Pr\limits_{f_1,f_2}\left[\left\langle (P^*_1(G,f_1),P^*_2(G,f_2)), V\right\rangle=1\right]$$ for a malicious prover without knowledge of a witness. 

\begin{lemma}  \label{lma:kappa}
    The maximum probability with which any dishonest provers without knowledge of a witness can pass the check phase is $1-\frac{1}{3|H|}$. 
\end{lemma}
The proof of this lemma is similar to the proof of the classical soundness property of Protocol \ref{prot:3_prot_party} already proven in \cite{crepeau2019practical}. However, here we derive a tighter bound. We refer to the proof of \Cref{lma:kappa1} in \Cref{app:CPoK-2-party-3COL} for the detailed proof of this lemma.

In \Cref{thm:clas_3col_ext}, we use \Cref{lma:kappa} to prove a lower bound on the knowledge error of the $2$-party protocol against the classical provers.

\begin{theorem}[Classical PoK of Protocol \ref{prot:3_prot_party}]
\label{thm:clas_3col_ext}
    $\langle (P^*_1,P^*_2),V\rangle$ is a classical proof of knowledge for 3-coloring with knowledge error $1-\frac{1}{3|H|}$.
\end{theorem}
\begin{proof}
   We consider random classical strategies where the provers share an internal random variable denoted by $r$. We suppose that  the acceptance probability of the verifier $p_{acc}$, averaged over all possible values of $r$ (same as averaged over all possible random strategies $f_1,f_2$), is given to be strictly greater than $\left(1 - \frac{1}{3|H|}\right)$: 
\begin{align}
p_{acc}:=\Pr_{f_1,f_2}\left[\left\langle (P^*_1(G,f_1),P^*_2(G,f_2)), V\right\rangle=1\right] \geq \left(1 - \frac{1}{3|H|}\right).
\end{align}
Note that during the rewinding process, the internal random variable $r$ is fixed, which effectively means we are now working with deterministic strategies for the provers. We now assume the fixed internal random variable is $r$, and during the rewinding proces, the provers utilize deterministic strategies fixed under this $r$.

With a deterministic strategy fixed by the choice of $r$, if the acceptance probability for the verifier is strictly greater than $\left(1 - \frac{1}{3|H|}\right)$, then by \Cref{lma:kappa}, the extractor can successfully extract the correct coloring with probability $1$. We denote the set of all such random coins (values of $r$) where the extractor succeeds with certainty as $R_{win}$. For the complement set, denoted $R_{fail}$, the extractor succeeds with a probability of less than $1$. Therefore, the overall probability that the extractor successfully extracts a correct witness is at least $\sum_{r \in R_{win}} p_r$, where $p_r$ is the acceptance probability associated with the random coin $r$.

To prove classical PoK (similar as \Cref{def:QPoK}), we need to show $\sum_{r\in R_{win}} p_r \geq \frac{1}{poly(\eta)}(p_{acc} - \kappa)$, where $\kappa$ is the knowledge error. From previous definition of $R_{win}$ and $R_{fail}$, for all the random coins $r \in R_{fail}$, the verifier will accept provers' answers with probability at most $\left(1 - \frac{1}{3|H|}\right)$. We also assume in the beginning that the overall acceptance probability $p_{acc}$ is strictly greater than $\left(1 - \frac{1}{3|H|}\right)$, i.e., $p_{acc} = \left(1 - \frac{1}{3|H|} + \delta\right)$, for some $\delta >0$. Therefore, we can get the following equation:
\begin{align}
p_{acc} = \sum_{r\in R_{win}} p_r p_{acc|r} + \sum_{r\in R_{fail}} p_r p_{acc|r} &= 1 - \frac{1}{3|H|} + \delta
\end{align}
As $p_{acc|r} \leq 1$,
\begin{align}
\sum_{r\in R_{win}} p_r + \sum_{r\in R_{fail}} p_r p_{acc|r} &\geq 1 - \frac{1}{3|H|} + \delta.
\end{align}
Moreover, for all $r\in R_{fail}$, we have $p_{acc|r} < \left(1 - \frac{1}{3|H|}\right)$ as stated before. It follows that
\begin{align}
\sum_{r\in R_{win}} p_r + \sum_{r\in R_{fail}} p_r \left(1 - \frac{1}{3|H|}\right) &\geq 1 - \frac{1}{3|H|} + \delta\\
\sum_{r\in R_{win}} p_r +  \left(1 - \frac{1}{3|H|} \right)\left(1 - \sum_{r\in R_{win}} p_r\right) &\geq 1 - \frac{1}{3|H|} + \delta\\
\frac{1}{3|H|}\sum_{r\in R_{win}} p_r &\geq \delta\\
\sum_{r \in R_{win}} p_r &\geq 3|H|\delta.
\end{align}

Setting the knowledge error $\kappa := 1 - \frac{1}{3|H|}$, we have $\delta = p_{acc} - \kappa$. Therefore, from the above inequality we get $\sum_{r \in R_{win}} p_r\geq 3|H|(p_{acc} - \kappa)$. This concludes the proof. \qed
\end{proof}

\subsubsection{Quantum Proofs of Knowledge for the 3-Coloring Problem}
We turn now to consider the construction of an extractor and bounding the knowledge error for a proof system secure against quantum provers. We describe in Protocol \ref{prot:3col} the interaction between $(\langle P_1, P_2, P_3 \rangle,V)$ as originally described in \cite{crepeau2019practical}.

\begin{protocol}{Interactive protocol for three-coloring $\Pi_{\textsf{qnl}}^{(3)}$: Three-prover, 3-COL\label{prot:3col}}
$P_1,P_2$ and $P_3$ pre-agree on a three-coloring $\sigma$ of graph $G = (\mathcal{V},H)$, and generate a random permutation $\pi$. They then compute $a_i = \pi(\sigma(i))$ for all vertices $i \in V$ and randomly select labellings $l_i^0$ and $l_i^1$ for $i\in V$ such that $l_i^0 + l_i^1 = a_i (\textsf{mod } 3)$ holds.
\begin{enumerate}
    \vspace{-1.5mm}
  \item \textbf{Commit phase.}
  \begin{enumerate}
    \item
    $V$ picks $(i,j,b), (i',j',b'), (i'',j'',b'')\in_{\mathcal{D}_G} ( H \times \mathbb{F}_2)^3$.
    \item
    $V_1$ sends $(i,j,b)$ to $P_1$, $V_2$ sends $(i',j',b')$ to $P_2$ and $V_3$ sends $(i'',j'',b'')$ to $P_3$.
    \item 
    If $(i,j) \in H$ and $b \in  \mathbb{F}_2$, then $P_1$ replies $l_i^b,l_j^b$.
    \item 
    If $(i',j') \in H$ and $b' \in  \mathbb{F}_2$, then $P_2$ replies $l_{i'}^{b'},l_{j'}^{b'}$.
    \item 
    If $(i'',j'') \in H$ and $b'' \in  \mathbb{F}_2$, then $P_3$ replies $l_{i''}^{b''},l_{j''}^{b''}$.
  \end{enumerate}
  \item \textbf{Check Phase}
    \begin{enumerate}
    \item (Consistency Test)
    \begin{itemize}
        \item If $(i,j)=(i'',j'')$ and $b=b''$, then $V$ accepts if $l_i^b=l_{i''}^{b''}\cap l_j^b= l_{j''}^{b''}$. 
        \item If $(i',j')=(i'',j'')$ and $b=b''$, then $V$ accepts if $l_{i'}^{b'}=l_{i''}^{b''}\cap l_{j'}^{b'}=l_{j''}^{b''}$. 
    \end{itemize}
    \item (Verification Test) 
    \begin{itemize}
        \item If $(i,j)=(i',j')$ and $b\neq b'$, then $V$ accepts iff $l_i^b + l_{i'}^{b'} \neq l_j^b + l_{j'}^{b'} $. 
    \end{itemize}
    \item (Well-Definition Test) 
    \begin{itemize}
        \item If $(i,j)=(i',j')$ and $b = b'$, then $V$ accepts iff $l_i^b = l_{i'}^{b'} \cap l_j^b = l_{j'}^{b'} $. 
        \item If $(i,j)\cap (i',j') = i$ and $b = b'$, then $V$ accepts iff $l_i^b = l_{i'}^{b'} $. 
        \item If $(i,j)\cap(i',j')=j$ and $b = b'$, then $V$ accepts iff $l_j^b = l_{j'}^{b'} $. 
    \end{itemize} 
  \end{enumerate}
\end{enumerate}
\end{protocol}

We construct an extractor $K$ for Protocol \ref{prot:3col} in Protocol \ref{prot:3_ext}. And we now give a formal proof of \Cref{thm:rzkp3col} bounding the knowledge error of the extractor in Protocol \ref{prot:3_ext}. However, to prove the main theorem we need more results. In Lemma \ref{lem:symmetric}, we show that the Protocol \ref{prot:3_ext} is a symmetric game.

\begin{lemma}
\label{lem:symmetric}
Protocol \ref{prot:3col} is a symmetric game.
\end{lemma}
\begin{proof}
The proof directly follows from the distribution $\mathcal{D}_G$.
\end{proof}

Suppose in Protocol \ref{prot:3_ext} the provers (possibly malicious) start with a tri-partite quantum state $\rho_{S_1S_2S_3}$ and $\Pr\big[\langle (P^*_1,P^*_2,P^*_3)(G,\rho_{S_1S_2S_3}) ,V\rangle$ $= 1\big]$ denotes the average probability that verifier accepts answers from the quantum provers.  

\begin{protocol}{Knowledge Extractor for Quantum Provers in $\Pi_{\textsf{qnl}}^{(3)}$\label{prot:3_ext}}
The extractor starts with two consistency check sets $f_1 = \emptyset$, $f_2 = \emptyset$, one witness set $w = \emptyset$, and two temporary edge sets $\tilde H_1 = H$, $\tilde H_2 = H$.
\begin{enumerate}
\item The extractor repeats the following until $\tilde H_1 = \emptyset$.
\begin{enumerate}
\item The extractor $K$ sends one edge $(u,v)\in \tilde H_1$, and $b\in \F_2$, to prover $P^*_1$.
\item $P^*_1$ replies with $l^{b}_u,l^{b}_v \in \F_3$.
\item It rewinds $P^*_1$ and sends the same edge $(u,v)$, and a different $b' \neq b\in \F^*_3$, to prover $P^*_1$.
\item $P^*_1$ replies with $l^{b'}_u,l^{b'}_v \in \F_3$.
\item It updates $f_1 = f_1 \cup \{(l^b_u,l^b_v,l^{b'}_u,l^{b'}_v)\}$.
\item It also updates $\tilde H_1 = \tilde E_1 \setminus \{(u,v)\}$.
\item The extractor rewinds $P^*_1$.
\end{enumerate}
\item The extractor repeats the following until $\tilde H_2 = \emptyset$.
\begin{enumerate}
\item The extractor $K$ sends one edge $(u,v)\in \tilde H_2$, and $b\in \F^*_3$, to prover $P^*_2$.
\item $P^*_2$ replies with $\tl^{b}_u,\tl^{b}_v \in \F_3$.
\item It rewinds $P^*_2$ and sends the same edge $(u,v)$, and a different $b'\in \F^*_3$ to prover $P^*_2$.
\item $P^*_2$ replies with $\tl^{b'}_u,\tl^{b'}_v \in \F_3$.
\item It updates $f_2 = f_2 \cup \{(\tl^b_u,\tl^b_v,\tl^{b'}_u,\tl^{b'}_v)\}$.
\item It also updates $\tilde H_2 = \tilde H_2 \setminus \{(u,v)\}$.
\item The extractor rewinds $P^*_2$.
\end{enumerate}
\item The extractor uses the $2$-party extractor $\tilde K$ with the classical strategies $f_1, f_2$ for the classical provers $\tP_1$, and $\tP_2$ respectively for extracting an witness $w$.
\end{enumerate}
\end{protocol}

\begin{lemma}
\label{lem:dist_q_c}
The absolute difference between the two probability distributions is bounded:
\begin{align} \label{eq:dist_c_q}
|\Pr\big[\langle  \left( P^*_1,P^*_2,P^*_3\right)(G,\rho_{S_1S_2S_3}), V\rangle = 1\big] -\Pr_{f_1,f_2}\big[\langle\tP_1(G,f_1),\tP_2(G,f_2),V\rangle = 1\big]| \leq \delta,
\end{align}
where $\delta := 16|H|\sqrt{1 - \Pr\big[\langle  \left( P^*_1,P^*_2,P^*_3\right)(G,\rho_{S_1S_2S_3}), V\rangle = 1\big]}$.
\end{lemma}
The proof of this lemma is similar to the of proof of Lemma $18$ in \cite{kempe2011entangled}. However, here we achieve a slightly tighter bound. We refer to the proof of \Cref{lem:dist_q_c1} in the appendix. In the next lemma, we give a lower bound on the $\delta$ parameter in \Cref{lem:dist_q_c}.
\begin{lemma} \label{lem:lower_bound}
For all $|H|\geq 3$, $\delta \geq \frac{16}{2401|H|^4}$, where $\delta$ is defined in \Cref{lem:dist_q_c}.
\end{lemma}
The proof of this lemma is provided in \Cref{sec:3col-proofs}, where this lemma is presented as \Cref{lem:lower-bound1}. Now, we give the proof of \Cref{thm:rzkp3col}, which is presented more formally as \Cref{thm:rzkp3col_form}.

\begin{theorem}
\label{thm:rzkp3col_form}
If there exists a canonical extractor $\tilde K$ for the $2$-party Protocol \ref{prot:3_prot_party} that can extract a witness with a classical knowledge error $\kappa_c:=1- \frac{1}{3|H|}$, then there exists a canonical extractor $K$ for the $3$-party Protocol \ref{prot:3_ext} that can extract a witness with a quantum knowledge error $\kappa_q$ with
\begin{equation}
\kappa_q \geq \kappa_c + \tilde \delta, 
\end{equation}
where $\tilde \delta \geq \frac{1}{3|H|} - \left(\frac{1}{7|H|}\right)^4$.
\end{theorem}

\begin{proof}
Since Protocol \ref{prot:3_ext} is symmetric, all the malicious provers use the same strategy. Therefore we write the distribution $\Pr\big[\langle (P^*_1(G,\rho_{S_1S_2S_3}),P^*_2(G,\rho_{S_1S_2S_3}), P^*_3(G,\rho_{S_1S_2S_3})), V\rangle = 1\big]$ as $$\Pr\big[\langle \left(P^*_1,P^*_2,P^*_3\right) (G,\rho_{S_1S_2S_3}) ,V\rangle=1\big].$$  For simplicity, we denote two acceptance probability as $p^c_{acc}$ and $p^q_{acc}$ for classical and quantum adversaries correspondingly:
\begin{align}
    p^c_{acc} &:=\Pr_{f_1,f_2}\left[\left\langle (P^*_1(G,f_1),P^*_2(G,f_2)), V\right\rangle=1\right]\\
    p^q_{acc} &:= \Pr_{\rho}\big[\langle \left(P^*_1,P^*_2,P^*_3\right)  (G,\rho_{S_1S_2S_3}) ,V\rangle = 1\big].
\end{align}

To prove the theorem, we assume $p^q_{acc} > \kappa_q$, and define $\kappa_q := \kappa_c + \tilde \delta$. Then it suffices to show that $p^q_{acc} > \kappa_q$ implies $p^c_{acc} > \kappa_c$ for some non-trivial value of $\kappa_q$. Once this statement is established, since a classical proof of knowledge is assumed, we obtain a classical extractor $\tilde K$ with success probability greater than $p^c_{acc} - \kappa_c$: 
\begin{align}
p^c_{acc} > \kappa_{c}  \Longrightarrow \Pr_{f_1,f_2}[(G,w) \in \R: w \leftarrow \tilde K^{P^*(G,w_{c_1},w_{c_2})}(G)] \geq p^c_{acc} - \kappa_c.
\end{align}

In \Cref{lem:dist_q_c}, we establish a connection between the probabilities $p_{acc}^c$ and $p_{acc}^q$. According to \Cref{eq:dist_c_q},  we get
\begin{align}
p_{acc}^q-p_{acc}^c \leq \tilde \delta
\Longrightarrow p_{acc}^c \geq p_{acc}^q- \tilde \delta.
\end{align}
Then assumption gives that
\begin{align}
p_{acc}^q  - \tilde \delta > \kappa_c 
\Longrightarrow p_{acc}^c\geq \kappa_c. 
\end{align} Now we successfully demonstrate that the quantum extractor in Protocol \ref{prot:3_ext} can successfully extract a witness with probability greater than $p^q_{acc} - \kappa_q$, thereby completing the proof.

We now focus on proving the suffice-to-prove statement: $p^q_{acc} > \kappa_q \Longrightarrow p^c_{acc} > \kappa_c$ for some non-trivial value of $\kappa_q$. From \Cref{lem:dist_q_c}, we have
\begin{equation}
\label{eq:pqdiff}
|p^q_{acc} - p^c_{acc}| \leq \delta,
\end{equation}
where $\delta := 16|H|\sqrt{1 - p^q_{acc}}.$ If $p^q_{acc} > \kappa_q$, it follows from \Cref{eq:pqdiff} that
\begin{align}
p^q_{acc} - p^c_{acc} &\leq \delta \\
p^c_{acc} &\geq p^q_{acc} - \delta\\
p^c_{acc} &> \kappa_c + \tilde\delta - \delta~~\text{as } p^q_{acc} > \kappa_q = \kappa_c + \tilde \delta.
\end{align}
Therefore, as long as $\tilde \delta - \delta >0$ we get $p^c_{acc} > \kappa_c$. 

Now, we choose a suitable value of $\tilde \delta$. We set $\kappa_q := 1 - \left(\frac{1}{7|H|}\right)^4$ and $\kappa_c := 1 - \frac{1}{3|H|}$ (coming from quantum and classical soundness parameter). Then it follows that
\begin{align}
\tilde \delta &:= \kappa_q - \kappa_c\\
&= \frac{1}{3|H|} - \left(\frac{1}{7|H|}\right)^4.
\end{align}
Similarly, since by assumption $p^q_{acc} > \kappa_q = 1 - \left(\frac{1}{7|H|}\right)^4$, then we get
\begin{align}
\delta &\leq 16|H|\sqrt{1 - \kappa_q}\\
&\leq \frac{16}{49|H|}.
\end{align}
The above two equations give us that $\tilde \delta - \delta >0$ as expected. This concludes the whole proof. \qed
\end{proof}

\section{Improved Soundness Error for Relativistic $\Sigma$-Protocols}
\label{sec:improv-sound}
In this section, we prove our new bound relating to the application of two consecutive measurement operators, which we then use to prove improved bounds on the soudness error for relativstic $\Sigma$-protocols.

\subsection{Tighter Lower Bound on the Two Consecutive Measurement Operators}
We now formally state \Cref{thm:consec_inf} and prove it using an analysis similar to that of Proposition $3$ in Appendix $C$ of \cite{chailloux2021relativistic}, where here, we require a bound only for two consecutive measurement operators rather than the three in \cite{chailloux2021relativistic}. In \Cref{sec:soundnessimprove}, we will use this theorem to improve the soundness error for the Hamiltonian cycle and subset sum relativistic ZKPs\cite{chailloux2017relativistic,crepeau2023zero}. 
\begin{theorem}
\label{thm:TwoOperators}
    
    Consider two sets of orthogonal projectors $\{W_1^{s_1}\}_{s_1\in S_1}$ and $\{W_2^{s_2}\}_{s_2\in S_2}$ with $W_i=\sum_{s_i=1}^{S_i}W_i^{s_i}$ and $W_i^sW_i^{s'}=\delta_{ss'}W_i^s$. 
     Let $F_1=\frac{1}{2}\left(\Tr[W_1\sigma]+\Tr[W_2\sigma]\right)$, $F_1>\frac{1}{2}$ and 
    \begin{equation*}
        F_2=\frac{1}{2}\sum_{s_1=1,s_2=1}^{\abs{S_1}\abs{S_2}}\left\{\Tr[W_2^{s_2}W_1^{s_1}\sigma W_1^{s_1}]+\Tr[W_1^{s_1}W_2^{s_2}\sigma W_2^{s_2}]\right\}.
    \end{equation*}
    Then, we have
    \begin{equation*}
        F_2\geq\frac{2}{\max_i\abs{S_i}}\left(F_1-\frac{1}{2}\right)^2.
    \end{equation*}
\end{theorem}
\begin{proof}
We first consider a purification $\ket{\Omega}$ of the state $\sigma$, and extend the projectors $W_i^{s_i}$ to operate on the purified state.
More concretely, let $A$ be the subsystem that the mixed state $\sigma_A$ acts on, we can add a subsystem $B$ to get a purification $\ket{\Omega}_{AB}$ of the mixed state, i.e. $\Tr_B[\ket{\Omega}_{AB}]=\sigma_A$.
The projectors $W_i$ acting on $A$ can be extended to act on $AB$ by including an identity operator over the subsystem $B$.
We first use the inequality of $\sum_{s_i=1}^{\abs{S_i}}W_i^{s_i}\dyad{\psi}W_i^{s_i}\geq\frac{1}{\abs{S_i}}W_i\dyad{\psi}W_i$ (Proposition $4$ in~\cite{chailloux2017relativistic}) to simplify,
\begin{align*}
    E&=\frac{1}{2}\left\{\sum_{s_1=1}^{\abs{S_1}}\Tr[W_2W_1^{s_1}\dyad{\Omega} W_1^{s_1}]+\sum_{s_2=1}^{\abs{S_2}}\Tr[W_1W_2^{s_2}\dyad{\Omega} W_2^{s_2}]\right\}\\
    &\geq\frac{1}{2}\left\{\frac{1}{\abs{S_1}}\Tr[W_2W_1\dyad{\Omega} W_1]+\frac{1}{\abs{S_2}}\Tr[W_1W_2\dyad{\Omega} W_2]\right\}\\
    &\geq\frac{1}{2\max_i \abs{S_i}}\left\{\Tr[W_2W_1\dyad{\Omega} W_1]+\Tr[W_1W_2\dyad{\Omega} W_2]\right\}.
\end{align*}
Let us now examine this simplified case, where there is only one outcome for each projector.\\

Let $\ket{\phi_i}=\frac{W_i\ket{\Omega}}{\norm{W_i\ket{\Omega}}_1}$, which allows us to expand the purified state as
\begin{equation*}
    \ket{\Omega}=\cos(\alpha_i)\ket{\phi_i}+\sin(\alpha_i)\ket{\phi_i^{\perp}},
\end{equation*}
with $\alpha_i\in[0,\frac{\pi}{2}]$ (both sine and cosine positive) in general.
Let us write
\begin{align*}
    \ket{\phi_2}=&\cos(\alpha_2)\ket{\Omega}+\sin(\alpha_2)\ket{B}\\
    \ket{\phi_1}=&\cos(\alpha_1)\ket{\Omega}+x\ket{A}+y\ket{B},
\end{align*}
with $\cos[2](\alpha_1)+\abs{x}^2+\abs{y}^2=1$, and $\{\ket{\Omega},\ket{A},\ket{B}\}$ form an orthonormal set.
We note that the first expression is valid since any complex coefficient present for $\ket{B}$ can be absorbed into $\ket{B}$, while the second expansion introduces an orthogonal $\ket{A}$ (e.g. via Gaussian elimination).
Since $W_2\ket{\phi_2}=\ket{\phi_2}$ using the definition of a projector ($W_2^2=W_2$), we can in general define 
\begin{equation*}
    W_2=\dyad{\phi_2}+W_2',
\end{equation*}
where $W_2'\ket{\phi_2}=0$.
Expanding $W_2\ket{\Omega}=\cos(\alpha_2)\ket{\phi_2}+W_2'\ket{\Omega}$, we immediately have that $W_2'\ket{\Omega}=0$ and consequently $W_2'\ket{B}=0$ (by expanding $W_2'\ket{\phi_2}$). 
We have thus
\begin{align*}
    W_2W_1\ket{\Omega}=&\cos(\alpha_1)W_2\ket{\phi_1}\\
    =&\cos(\alpha_1)[\braket{\phi_2}{\phi_1}\ket{\phi_2}+W_2'\ket{\phi_1}]\\
    =&\cos(\alpha_1)[(\cos\alpha_1\cos\alpha_2+y\sin\alpha_2)\ket{\phi_2}+xW_2'\ket{A}]\\
    =&\cos\alpha_1[(\cos\alpha_1\cos\alpha_2+y\sin\alpha_2)\ket{\phi_2}+z\ket{A'}],
\end{align*}
where we have $\abs{z}\leq\abs{x}$ and $\ket{A'}$ being orthogonal to $\ket{B}$ and $\ket{\Omega}$ (since $W_2'$ applied on these states give 0).
The first line applies $W_1$ on $\ket{\Omega}$, the second line expands $W_2$, the third line expands $\ket{\phi_1}$, while the final line defines the projection outcome of $xW_2'\ket{A}$.\\

Therefore, we can simplify the trace as
\begin{align*}
    \Tr[W_2W_1\dyad{\Omega}W_1]&=\norm{W_2W_1\ket{\Omega}}_1^2\\
    &=\cos^2\alpha_1\left[\abs{\cos\alpha_1\cos\alpha_2+y\sin\alpha_2}^2+\abs{z}^2\right]\\
    &\geq \cos^2\alpha_1\abs{\cos\alpha_1\cos\alpha_2+y\sin\alpha_2}^2,
\end{align*}
since $\ket{\phi_2}$ and $\ket{A'}$ are orthogonal and $\abs{z}^2\geq0$.
We note that in the expansion of $\ket{\phi_1}$, we have $\abs{y}^2\leq\sin^2\alpha_1$.
Since $\cos\alpha_1\cos\alpha_2$ is real and $y$ can be complex, labelling $y=re^{i\theta}$, the modulus can be interpreted as the distance from the origin for a vector which is a sum of a vector on the real line of length $\cos\alpha_1\cos\alpha_2$ and a vector which can lie anywhere on a circle with length $\abs{y}\sin\alpha_2$.
We can thus bound
\begin{align*}
    \abs{\cos\alpha_1\cos\alpha_2+y\sin\alpha_2}^2&\geq(\cos\alpha_1\cos\alpha_2-\abs{y}\sin\alpha_2)^2\\
    &\geq(\cos\alpha_1\cos\alpha_2-\sin\alpha_1\sin\alpha_2)^2\\
    &=\cos(\alpha_1+\alpha_2)^2,
\end{align*}
where the second inequality uses the fact that $\cos\alpha_1\cos\alpha_2-\sin\alpha_1\sin\alpha_2>0$ from the fact that $\alpha_1,\alpha_2\in \left[0,\frac{\Pi}{4}\right]$ and $F_1>\frac{1}{2}$.
Since $\cos(\alpha_1+\alpha_2)\geq\cos^2\alpha_1+\cos^2\alpha_2-1$, we have that
\begin{align*}
    \Tr[W_2W_1\dyad{\Omega}W_1]&\geq v_{min}\cos^2\alpha_1\\
    & \geq \cos^2(\alpha_1+\alpha_2)\cos^2(\alpha_1-\alpha_2)\cos^2\alpha_1\\
    & \geq (\cos^2\alpha_1+\cos^2\alpha_2-1)^2\cos^2\alpha_1.
\end{align*}
A similar argument can be made for the other trace term, with $W_1$ and $W_2$ switched, and by symmetry, we expect
\begin{equation*}
    \Tr[W_2W_1\dyad{\Omega}W_1]+\Tr[W_1W_2\dyad{\Omega}W_2]\geq(\cos^2\alpha_1+\cos^2\alpha_2-1)^2(\cos^2\alpha_1+\cos^2\alpha_2)
\end{equation*}
We note that we can expand $F_1=\frac{1}{2}(\cos^2\alpha_1+\cos^2\alpha_2)$ and define $F_1=\frac{1}{2}+\epsilon$, giving us
\begin{equation*}
    \Tr[W_2W_1\dyad{\Omega}W_1]+\Tr[W_1W_2\dyad{\Omega}W_2]\geq(2\epsilon)^2(1+2\epsilon)\geq 4\epsilon^2.
\end{equation*}
As such, combining the results, we have
\begin{align*}
    F_2&\geq\frac{1}{2\max_i \abs{S_i}}\left\{\Tr[W_2W_1\dyad{\Omega} W_1]+\Tr[W_1W_2\dyad{\Omega} W_2]\right\}\\
    &\geq\frac{1}{2\max_i \abs{S_i}}\times 4\epsilon^2\\
    &=\frac{2}{\max_i \abs{S_i}}\left(F_1-\frac{1}{2}\right)^2.
\end{align*}
\end{proof}

\subsection{Two Choice Coupling Games}

Here, we consider entangled games $G^2$ involving two parties, with the first party receiving input $x\in X$, while the second party receives one of two possible inputs, $y\in\{0,1\}$.
Without communicating with each other, the parties respond with $a\in A$ and $b\in B$ respectively, and win when $V(a,b,x,y)=1$.
We assume that $x$ and $y$ are selected randomly.
We define the coupling game $G^2_{coup}$ with the same setup, except that Bob would be asked additionally $\bar{y}$, and responds with $b'\in B$.
The winning condition would be then $V(a,b,x,y)=V(a,b',x,\bar{y})=1$.\\

The coupling game is closely related to the original game, as shown in the theorem below.
We can thus bound the winning probability of the original game in the soundness analysis by bounding the winning probability of the coupling game.
\begin{theorem}
\label{thm:TwoChoiceCoupling}
    Consider an entangled game with two possible inputs for Bob, $G^2$ with $\omega^*(G^2)>\frac{1}{2}$. Then,
    \begin{equation*}
        \omega^*(G_{coup}^2)\geq \frac{2}{\abs{S_{max}}} \left(\omega^*(G^2)-\frac{1}{2}\right)^2,
    \end{equation*}
    where $\omega^*(G^2)$ is the winning probability of $G^2$, $\omega^*(G_{coup}^2)$ is the winning probability of the corresponding coupling game, and $\abs{S_{max}}$ is the maximum number of accepted responses from Bob, i.e. $\abs{S_{max}}=\max_{a,x,y}\abs{\{b:V(a,b,x,y)=1\}}$.
\end{theorem}
The proof of the above statement follows closely the analysis given in \cite{chailloux2017relativistic}, and so we defer it to \Cref{sec:app-coupling}

\subsection{Improved Bounds}
\label{sec:soundnessimprove}
We use \Cref{thm:TwoChoiceCoupling} to arrive at the following theorems using analyses similar to \cite{chailloux2017relativistic} and \cite{crepeau2023zero}, respectively, and we defer the proofs to \Cref{sec:chai-sound} and \Cref{sec:crep-sound} respectively.
\begin{theorem}
\label{thm:chai-sound}
The $\F_Q$-commitment based two-prover relativistic zero-knowledge proof protocol for the Hamiltonian cycle that is proposed in~\cite{chailloux2017relativistic} (See Protocol \ref{protocol:hamiltonian-cycle}) has soundness error $$\epsilon_{sound} \leq \frac{1}{2} + \left(\frac{n!}{2Q}\right)^{\frac{1}{2}},$$
where $n$ is the number of nodes in the underlying graph.
\end{theorem}

We can have a similar analysis with the subset-sum ZKP in~\cite{crepeau2023zero},
\begin{theorem}
\label{thm:crep-sound}
The $\F_Q$-commitment based two-prover relativistic zero-knowledge proof protocol for the subset sum that is proposed in~\cite{crepeau2023zero} (See Protocol \ref{protocol:subset-sum})  has soundness error $$\epsilon_{sound} \leq \frac{1}{2} + \left(\frac{2^{n-1}}{Q}\right)^{\frac{1}{2}},$$
where $n$ is the size of the set.
\end{theorem}

\section*{Acknowledgments}
    The authors thank Ziyue Xin and Claude Crepeau for identifying and bringing attention to issues in the proofs.

\section*{Disclaimer}

This paper was prepared for informational purposes by the Global Technology Applied Research center of JPMorgan Chase \& Co. This paper is not a product of the Research Department of JPMorgan Chase \& Co. or its affiliates. Neither JPMorgan Chase \& Co. nor any of its affiliates makes any explicit or implied representation or warranty and none of them accept any liability in connection with this paper, including, without limitation, with respect to the completeness, accuracy, or reliability of the information contained herein and the potential legal, compliance, tax, or accounting effects thereof. This document is not intended as investment research or investment advice, or as a recommendation, offer, or solicitation for the purchase or sale of any security, financial instrument, financial product or service, or to be used in any way for evaluating the merits of participating in any transaction.

\bibliography{submission/bibliography}
\bibliographystyle{ieeetr}

\newpage
\appendix
\section{Appendix: Exemplary Relativistic $\Sigma$-Protocol}\label{app:ex_sigma}
In this section, we describe relativistic zero-knowledge proof systems for Hamiltonian Cycle and Subset-Sum from the literature, and use our results from \Cref{sec:QPoKswithSpec} to prove that they are indeed QPoKs. 

\subsection{Hamiltonian Cycle~\cite{chailloux2017relativistic}}

\begin{protocol}{Relativistic $\Sigma$-protocol for Hamiltonian Cycle~\cite{chailloux2017relativistic}\label{protocol:hamiltonian-cycle}}
\vspace{0.2mm}
$P_1$ and $P_2$ pre-agree on a Hamiltonian cycle $\mathcal{C}$ of a given graph $G=(\mathcal{V},H)$. They also agree beforehand on a random permutation $\Pi: \mathcal{V} \rightarrow \mathcal{V}$ and on an $n \times n$ matrix $A \in \mathcal{M}_n^{\mathbb{F}_Q}$ where each element of A is chosen uniformly at random in $\mathbb{F}_Q$.
\begin{enumerate}
    \item \text{[}Commitment to each bit of $M_{\Pi(G)}$.\text{]} $V$ sends a matrix $B \in \mathcal{M}_n^{\mathbb{F}_Q}$ to $P_1$, where each element of $B$ is chosen uniformly at random in $\mathbb{F}_Q.$ $ P_1$ outputs the matrix $Y \in \mathcal{M}_n^{\mathbb{F}_Q}$ such that $\forall i, j \in [n], Y_{i, j}=$ $A_{i, j}+\left(B_{i, j} *\left(M_{\Pi(G)}\right)_{i, j}\right)$.
    \item \text{[}Challenge phase.\text{]} $V$ sends a random bit $ ch  \in\{0,1\}$ to the $P_2$.
    \begin{itemize}
        \item If $ch=0, P_2$ decommits to all the elements of $M_{\Pi(G)}$, i.e. he sends all the elements of $A$ to $V_2$ and reveals $\Pi$.
        \item If $ch=1, P_2$ reveals only the bits (of value 1) of the adjacency matrix that correspond to a Hamiltonian cycle $\mathcal{C'}$ of $\Pi(G)$, i.e. for all edges $(u, v)$ of $\mathcal{C'}$, he sends $A_{u, v}$ as well as $\mathcal{C'}$.
    \end{itemize}
    \item \text{[}Check phase.\text{]} $V$ checks that those decommitments are valid and correspond to what $P_1$ has declared. $V$ also checks that the timing constraint of the bit commitment is satisfied. This means that
    \begin{itemize}
        \item If $ch=0$, $P_2$'s opening $A$ must satisfy $\forall i, j \in[n], Y_{i, j}=A_{i, j}+\left(B_{i, j} *\left(M_{\Pi(G)}\right)_{i, j}\right)$.
        \item If $ch=1$, $P_2$'s opening $A$ must satisfy $\forall(u, v) \in \mathcal{C'}, Y_{u, v}=A_{u, v}+B_{u, v}$.
    \end{itemize}
\end{enumerate}
\end{protocol}

\begin{theorem}
    The above proof system $\langle (P_1,P_2),V\rangle$ is a zero-knowledge QPoK for Hamiltonian Cycle with knowledge error $\frac{1}{2}+\xi$ for some $\xi$ negligible in $\eta$. 
\end{theorem}
\begin{proof}
    The protocol is a relativistic $\Sigma$-protocol based on \Cref{def:rel-sigma-protocol}. 

     \medskip\noindent \textsf{Special soundness.} Fix any $rand$, $com$, and let $(rand,com, ch, resp)$ and $(rand,com,ch',resp')$ be two accepting conversations for a given graph $G$ with $ch \neq ch'$. Without loss of generality, let $ch=0$ and $ch'=1$. Here, $rand$ corresponds to a matrix $B$ and $com$ corresponds to the output $Y$. For $ch=0$, the response $resp$ is a permutation $\Pi$ and the matrix $A$. Since it is an accepting conversation, we have 
    $$\forall (i,j), A_{i,j} = Y_{i,j} - B_{i,j} * (M_{\Pi(G)})_{i,j}.$$ For $ch'=1$, the response $resp'$ is a cycle $C'$ on the permuted graph $\Pi(G)$ and corresponding $A'_{u,v}$ for all edges $(u,v) \in C'$. Since this is also an accepting conversation, the openings satisfy $$\forall(u,v)\in C', A'_{u,v} = Y_{u,v} - B_{u,v}.$$      

    We first note that $C'$ is a Hamiltonian cycle of $\Pi(G)$ iff $\Pi^{-1}(C')$ is a Hamiltonian cycle of the graph $G$, since $\Pi$ is a permutation matrix. As such let us consider an extractor that extracts witness $$w :=K_{0}(G,rand,com,ch,resp,ch',resp') := \Pi^{-1}(C').$$The probability that the extractor is successful is 
    \begin{align*}
        \Pr[(w,x)\in R|rand,com,acc]=&\Pr[C'\in HC(\Pi(G))|rand,com,acc]\\
        =&1-\Pr[C'\not\in HC(\Pi(G))|rand,com,acc],
    \end{align*}
    where $acc$ represents that both conversations are accepted, $HC(\Pi(G))$ represents the set of Hamiltonian cycles in graph $\Pi(G)$, and the probability is evaluated over the choice of $ch$, $ch'$, $resp$ and $resp'$ for accepting conversations. If the cycle $C'$ provided is not a Hamiltonian cycle, there exists an edge $(u,v)\in C'$ such that $\left(M_{\Pi(G)}\right)_{u,v}=0$. Conditioned on the acceptance of both conversations, we have that $$A_{u,v} = Y_{u,v} \text{ and } A'_{u,v} = Y_{u,v} - B_{u,v},$$which further implies $A_{u,v}-A'_{u,v} = B_{u,v}$. Since $A_{u,v}$ and $A'_{u,v}$ are $P_2$'s output while $B_{u,v}$ is the input of $P_1$, the no-signalling constraint ensures that the probability of this event occurring is upper bounded by $\frac{1}{Q}$, the probability of $P_2$ guessing $B_{u,v}$~\cite{chailloux2017relativistic}. As such,
    \begin{align*}
        \Pr[(w,x)\in R|rand,com,acc]\geq&1-\frac{1}{p_{acc}}\Pr[A_{u,v}-A'_{u,v} = B_{u,v}|rand,com]\\
        \geq&1-\frac{1}{Qp_{acc}},
    \end{align*}
    which gives a special soundness of $\frac{1}{Qp_{acc}}$. Further assuming that the winning probability $\Pr[V\,accepts]>\frac{1}{c}+\xi$ for some negligible $\xi$,  the special soundness parameter is negligible, $\delta_{SS}<\frac{1}{p(\eta)}$ (using methods from the proof of \Cref{prop:SpecSound-HC}).
    
    \medskip\noindent \textsf{Extractability.} Since $\langle (P_1,P_2),V\rangle$ has negligible special soundness when $\Pr[V\,accepts]>\frac{1}{c}+\xi$, and $\# C_{\eta x} = 2$, by \Cref{thm:sig-QPoK}, we have that $\langle (P_1,P_2),V\rangle$ is extractable with knowledge error $\frac{1}{2}+\xi$. 

    \medskip\noindent \textsf{Zero-knowledge.} It has been proven in~\cite{chailloux2017relativistic} that this protocol is perfectly zero-knowledge against quantum adversaries.
\end{proof}

\subsection{Subset Sum~\cite{crepeau2023zero}}

\begin{protocol}{Relativistic $\Sigma$-protocol for Subset Sum\cite{crepeau2023zero} \label{protocol:subset-sum}} 
\vspace{0.2mm}
$P_1$ and $P_2$ pre-agree on a solution $v$ of the given subset problem $(s,k) \in \mathbb{F}^n_Q \times \mathbb{F}_Q $, where $s$ is the input set and $k$ is the subset sum target. They also agree beforehand on random vectors $c_0,c_1 \in \mathbb{F}^n_Q$ and $z \in \mathbb{F}^n_2$.
\begin{enumerate}
    \item \text{[}Commitment to $s*z$ and $s*\bar{z}$.\text{]} $V$ sends a random value $a \in \mathbb{F}_Q$. $P_1$ replies with $w_0=a \cdot(s * z)+c_0$ and $w_1=a \cdot(s * \bar{z})+c_1$.
    \item \text{[}Challenge phase.\text{]} $V$ sends $ch \in\{0,1\}$ to $P_2$.
    \begin{itemize}
        \item If $ch=0$, $P_2$ sends back $z, c_0, c_1$. 
        \item If $ch=1$, $P_2$ sends back the binary vector $x=v \oplus z$ and the value $c'=\sum_{i=1}^n\left(c_{x_i}\right)_i$.
    \end{itemize}
    \item \text{[}Check phase.\text{]}
    \begin{itemize}
        \item If $ch=0$, $V$ checks that $w_0=a \cdot(s * z)+c_0$ and $w_1=a \cdot(s * \bar{z})+c_1$. 
        \item If $ch=1$, $V$ checks that $\sum_{i=1}^n\left(w_{x_i}\right)_i=a k+c'$. 
    \end{itemize}
\end{enumerate}
\end{protocol}

\begin{theorem}
    The above proof system $\langle (P_1,P_2),V\rangle$ is a zero-knowledge QPoK for Subset-Sum with knowledge error $\frac{1}{2}+\xi$ for some negligible $\xi$. 
\end{theorem}
\begin{proof}
    The protocol is a relativistic $\Sigma$-protocol based on \Cref{def:rel-sigma-protocol}. 

    \medskip\noindent \textsf{Special soundness.} Let $(com, ch, resp)$ and $(com,ch',resp')$ be two accepting conversations for $x$ with $ch \neq ch'$. Without loss of generality, $ch=0$ and $ch'=1$. Then $resp$ contains $z,c_0,c_1$. Since it is an accepting conversation, we have 
    \begin{align*}
        w_0 &= a \cdot s*z +c_0 \text{ and } w_1 = a\cdot s *\bar{z} + c_1\\
        &\Longrightarrow w_b = a (\bar{b} \cdot s * z + b \cdot s * \bar{z}) + c_b.
    \end{align*} $resp'$ contains a binary vector $x$ and the value $c'$. Since this is also an accepting conversation, the openings satisfy $$\sum_{i=1}^n (w_{x_i})_i = ak + c'.$$ Combining these two equations, we get 
    \begin{align*}
    ak+c'=\sum_{i=1}^n\left(w_{x_i}\right)_i=\sum_{i=1}^n a \cdot\left(\overline{x_i} s_i z_i+x_i s_i \overline{z_i}\right)+\left(c_{x_i}\right)_i \\
    \Longrightarrow a\left(\sum_{i=1}^n\left(\bar{x_i} s_i z_i+x_i s_i \bar{z_i}\right)-k\right)=c'-\sum_{i=1}^n\left(c_{x_i}\right)_i .
    \end{align*}
    We now consider a candidate solution $w := x \oplus z$, where $x$ comes from $resp$ and $z$ comes from $resp'$. Then the corresponding sum is $w* s$,
    \begin{align*}
        \sum_{i=1}^n w_is_i = \sum_{i=1}^n (x \oplus z)_i s_i = \sum_{i=1}^n (\bar{x_i}z_i + x_i \bar{z_i})s_i =  \sum_{i=1}^n\left(\bar{x_i} s_i z_i+x_i s_i \bar{z_i}\right)
    \end{align*}
    Assume the subset does not sum to $k$, i.e. $s \not \in \text{Subset-Sum}$, then we have the above summation $\neq k$. This implies that 
    \begin{align*}
        a=\frac{c'-\sum_{i=1}^n\left(c_{x_i}\right)_i}{\left(\sum_{i=1}^n\left(\bar{x_i} s_i z_i+x_i s_i \bar{z_i}\right)-k\right)}.
    \end{align*}
    However, the left-hand value $a$ is uniformly randomly selected by the $V$ and sent to $P_1$ while the right hand values $c_0,c_1,c',x,z$ are all determined by $P_2$. Therefore by no-signalling, this equation holds with probability $\frac{1}{Q}$. Then with probability $1-\frac{1}{Q}$, we would have that $w * s$ sum to the target value $k$, i.e. $w * s = k$. And so $w$ is indeed a solution, i.e. $(s,w) \in \text{Subset-Sum}$.

    \medskip\noindent \textsf{Extractibility.} Similarly as described before.
\end{proof}

\section{PoK for the Graph $3$-Coloring Protocol}\label{sec:3col-proofs} 

\subsection{Classical PoK for $2$-Party Protocol \ref{prot:3_prot_party}} \label{app:CPoK-2-party-3COL}

We prove the following lemma, given in the maintext as \Cref{lma:kappa} on the knowledge error of the described extractor. 
\begin{lemma}  \label{lma:kappa1}
    The maximum probability with which any dishonest provers without knowledge of a witness can pass the check phase is $1-\frac{1}{3|H|}$.
\end{lemma}
\begin{proof}
We first argue that we only need to consider provers with deterministic strategy. 
Before the protocol starts, the provers will agree on some strategy $f^{R_1}_1,f^{R_2}_2$ with associated randomness $R=(R_1,R_2)$. Therefore the winning probability (probability of passing the check phase) would be 
\begin{align*}
    P_{\text{win}} = \max_{f^{R_1}_1,f^{R_2}_2} P_{\text{win}}^{f_R} 
\end{align*}

Define $V\left(f_1^{R_1},f_2^{R_2}\right)$ be an indicator function indicating whether the check phase passes $V\left(f_1^{R_1},f_2^{R_2}\right)=1$ if it passes). The winning probability with a fixed strategy can be written as
\begin{align*}
    P_{\text{win}}^{f_R} &:= \sum_{e,e'}P_{e,e'}\sum_{b,b'}P_{b,b'}\sum_{R_1,R_2} P_{R_1,R_2}V\left(f_1^{R_1},f_2^{R_2}|e,e',b,b'\right)\\
    &=\sum_{R_1,R_2}P_{R_1,R_2}\sum_{e,e'}P_{e,e'}\sum_{b,b'}P_{b,b'}V\left(f_1^{R_1},f_2^{R_2}|e,e',b,b'\right)\\
    &:= \sum_{R_1,R_2}P_{R_1,R_2} P_{\text{win}|R_1,R_2}^{f_R}\\
    &\leq \max_{R_1,R_2}\left\{ P_{\text{win}|R_1,R_2}^{f_R} \right\}
\end{align*}

It follows that we only need to consider some fixed deterministic strategy $f = \left(f_1,f_2\right)$ for provers, i.e. 
\begin{align*}
    P_{\text{win}} &\leq \max_{f^{R_1}_1,f^{R_2}_2,R_1,R_2} P_{\text{win}|R_1,R_2}^{f_R} \\
    &= \max_{f_1,f_2} P^f_{\text{win}}.
\end{align*}
    Consider an adversary where $P^*_1,P^*_2$ agree on some table of assignments of vertices which is not a witness, i.e. not a valid 3-coloring. This means there exists either a vertex $i$ with undefined color or an edge $(i,j)\in H$ with vertices that are assigned the same color.

    Let us consider the two conditions for the check phase to pass:
    \begin{itemize}
        \item Well-definition test is chosen and passed. According to the distribution of $\mathcal{D}_G$ in Equation \eqref{eq:wdt_dist}, $V$ chooses an edge $(i,j) \in H$ with probability $\frac{2}{3|H|}$, $b \in \mathbb{F}_2$ uniformly at random, and chooses the second edge $(i',j')$ uniformly at random from $\text{ngbr}(i)$ with probability $\frac{1}{2}$ and from $\text{ngbr}(j)$ with probability $\frac{1}{2}$. It also sets $b' = b$. Therefore, with probability $\frac{1-1/3}{|H|}= \frac{2}{3|H|}$, $(i,j)$ containing bad vertex $i$ is chosen. There exists $\frac{1}{2}$ probability that $i$ is chosen from $\{i,j\}$ to be the common vertex. Therefore $V$ accepts with probability at most $1-\frac{1}{3|H|}$.
        \item Edge verification test is chosen and passed. In this case, with probability $\frac{1}{3|H|}$, $V$ chooses an edge $(i,j)$, and sets $(i',j')=(i,j)$, $b' \neq b$. The probability of the bad edge $(i,j)$ being chosen is $\frac{1}{|H|}$ and therefore $V$ accepts with probability at most $1-\frac{1}{3|H|}$. 
    \end{itemize}
    Then it follows that the maximum probability with which the dishonest $P$ can pass the check phase is $1- \frac{1}{3|H|}.$
\end{proof}

\subsection{Proofs Related to the QPoK for the $3$-Prover Protocol \ref{prot:3col}}
We restate and give a proof of \Cref{lem:dist_q_c} from the main-text.
\begin{lemma}
\label{lem:dist_q_c1}
The absolute difference between the two probability distributions is bounded:  
\begin{align} \label{eq:dist_c_q1}
|\Pr\big[\langle (P^*_1,P^*_2,P^*_3)(G,\rho_{S_1S_2S_3}),V\rangle = 1\big] -\Pr_{f_1,f_2}\big[\langle(\tP_1(G,f_1),\tP_2(G,f_2)),V\rangle = 1\big]| \leq \delta,
\end{align}
where $\delta := 16|H|\sqrt{1 - \Pr\big[\langle (P^*_1,P^*_2,P^*_3)(G,\rho_{S_1S_2S_3}), V\rangle = 1\big]}$.
\end{lemma}

\begin{proof}[Proof of \Cref{lem:dist_q_c}]
The expression $\Pr\limits_{f_1,f_2}\big[\langle(\tP_1(G,f_1),\tP_2(G,f_2)),V\rangle = 1\big]$ denotes the average probability that the two classical provers $\tP_1,\tP_2$ without any quantum device convinces the verifier using the random strategies $f_1,f_2$ that they have a coloring assignment of the graph $G = (\Vt, H)$. According to Protocol \ref{prot:3_ext}, each of the sets $f_1, f_2$ consist of answers to the $2|H|$ possible questions. We index every question for $f_1$ with $q_1, \ldots q_N$, and index the question for $f_2$ as $q'_1, \ldots , q'_N$, where $N =2|H|$, and $q_i := (e_i,b_i) \in H \times \F_2$, $q'_i := (e_i,b'_i) \in H\times \F_2$ for all $1\leq i \leq N$. We denote the answer corresponding to question $q_i =(e_i:=(u_i,v_i),b_i)$ ($q'_i=(e_i:=(u_i,v_i),b'_i)$) as $a_i = (l_{b_i}^{u_i},l_{b_i}^{v_i})$ ($a'_i = (l_{b_i'}^{u_i},l_{b_i'}^{v_i})$). We denote the joint distribution of the random variables $f_1,f_2$ as $D(f_1;f_2) := D(a_1,\ldots , a_N;a'_1, \ldots , a'_N)$. Therefore, we get the following expression for $\Pr\limits_{f_1,f_2}\big[\langle (\tP_1(G,f_1),\tP_2(G,f_2)),V\rangle = 1\big]$.

\begin{align}\nonumber
&\Pr\limits_{f_1,f_2}\big[\langle (\tP_1(G,f_1),\tP_2(G,f_2)),V\rangle = 1\big] \\ \nonumber
=& \sum_{f_1,f_2} D(f_1,w_{c_{2}}) \sum_{q,q' \in H\times \F_2} \mathcal{D}_G(q,q') \Pr\big[\langle (\tP_1(G,f_1),\tP_2(G,f_2)),V\rangle = 1|f_1,f_2\big]\\ \nonumber
= &\sum_{q,q' \in H\times \F_2} \mathcal{D}_G(q,q') \sum_{f_1,f_2} D(f_1,w_{c_{2}}) \Pr\big[\langle (\tP_1(G,f_1),\tP_2(G,f_2)),V\rangle = 1|f_1,f_2\big]\\ \label{eq:dist_avg}
= &\sum_{q,q' \in H\times \F_2} \mathcal{D}_G(q,q') \sum_{\substack{f_1\setminus \{a\},\\ f_2\setminus \{a'\}}} \sum_{a,a'} D(f_1,w_{c_{2}}) \Pr\big[a,a'|f_1,f_2\big] \V(a,a',q,q').
\end{align}

We get the second equality just by interchanging the summation over $f_1,f_2$ with the summation over $q,q'$. For the third equality, first we assume that in $f_1$ ($f_2$) the answer for $q$ ($q'$) is $a$ ($a'$), and divide the sum over $f_1,f_2$ into two parts $f_1\setminus \{a\}, f_2\setminus \{a'\}$ and $a,a'$. Secondly, we define Boolean function $Ver(a,a',q,q')$ as follows.
\begin{equation}
\label{eq:ver_3col}
Ver(a,a',q,q') := \begin{cases}
0 ~~\text{If Verification Test or Well-Definition Test fails,} \\
1 ~~\text{Otherwise}.
\end{cases}
\end{equation}
We get the third equality in Equation \eqref{eq:dist_avg} by rewriting $\Pr\big[\langle (\tP_1(G,f_1),\tP_2(G,f_2)),V\rangle = 1|f_1,f_2\big]$ as $$\Pr\big[a,a'|f_1,f_2\big] \V(a,a',q,q').$$

Suppose, for a fixed $q,q'$, $p_c(a,a'|q,q')$ denotes the marginal distribution of $a,a'$. Therefore,
\begin{equation}
\label{eq:pclass}
p_c(a,a'|q,q') := \sum_{\substack{f_1 \setminus \{a\} \\ f_2\setminus \{a'\}}} D(f_1;f_2).
\end{equation}
Note that, using the expression in Equation \eqref{eq:pclass} we get the following expression of Equation \eqref{eq:dist_avg}.
\begin{align}
\label{eq:pclass_dist} \nonumber
&\Pr\limits_{f_1,f_2}\big[\langle(\tP_1(G,f_1),\tP_2(G,f_2)),V\rangle = 1\big] = \\
&\sum_{q,q' \in H\times \F_2}  \mathcal{D}_G(q,q') \sum_{a,a'} p_c(a,a'|q,q') \V(a,a',q,q').
\end{align}

Suppose at the beginning of the Canonical Extractor Protocol \ref{prot:3_ext}, the provers share a tripartite state $\rho_{S_{1}S_{2}S_{3}}:= |\Psi\rangle_{S_{1}S_{2}S_{3}}\langle \Psi|$, where $|\Psi\rangle_{S_{1}S_{2}S_{3}}$ is a pure state. From \Cref{lem:symmetric}, we get that Protocol \ref{prot:3col} is a symmetric protocol. Therefore without any loss of generality, we assume that the strategies of all the provers use same set of strategies $W :=\{W_{q}\}_{q \in H\times \F_2}$. Moreover, we also assume that for a fixed $e,b$ the strategy $W_{q}$ is a collection of projectors, i.e., $W_{q}:= \left\{W_{q}^{a}\right \}_{a \in \F_3\times \F_3}$, and $(W_{q}^{a})^2 = W_{q}^{a}$, $W_{q}^{a} \geq 0$, $\sum_{a \in \F_3\times \F_3}W_{q}^{a} = \id$. Therefore, we get the following expression of $\Pr\big[\langle (P^*_1,P^*_2,P^*_3)(G,\rho_{S_1S_2S_3}), V\rangle = 1\big]$.
\begin{align}\nonumber
\Pr\big[&\langle (P^*_1,P^*_2,P^*_3)(G,\rho_{S_1S_2S_3}), V\rangle = 1\big] = \\ \label{eq:pq_acc}
&\sum_{q,q',q''} \sum_{a,a',a''}p_{quant}(a,a',a''|q,q',q'')\V'(a,a',a'',q,q',q''),
\end{align}
where $\V'(a,a',a'',q,q',q'')$ is defined in Equation \eqref{eq:ver_3_party}, and we define $p_{quant}(a,a',a''|q,q',q'')$ in Equation \eqref{eq:pq}.
\begin{equation}
\label{eq:ver_3_party}
\V'(a,a',a'',q,q',q'') := \begin{cases}
0 ~~~\text{If any of the tests Fails in Protocol \ref{prot:3col},}\\
1 ~~~\text{Otherwise.}
\end{cases}
\end{equation}

\begin{equation}
\label{eq:pq}
p_{quant}(a,a',a''|q,q',q'') := \Tr[(W_q^a \otimes W_{q'}^{a'} \otimes W_{q''}^{a''}) \rho_{S_1S_2S_3}].
\end{equation}
From Equation \eqref{eq:pq}, by taking the marginal distribution on $a,a'$ we get,

\begin{equation}
\label{eq:aa}
p_{quant}(a,a'|q,q',q'') = \Tr[(W_q^a \otimes W_{q'}^{a'} \otimes \id) \rho_{S_1S_2S_3}].
\end{equation}
Note that due to non-signalling constraint, for any two $q'',q_1''$ the following holds.
\begin{equation}
p_{quant}(a,a'|q,q',q'') = p_{quant}(a,a'|q,q',q_1'').
\end{equation}
Therefore from now on we denote $p_{quant}(a,a'|q,q',q'')$ as $p_{quant}(a,a'|q,q')$. Following Equation \eqref{eq:aa} we get,

\begin{equation}
\label{eq:aaqq}
p_{quant}(a,a'|q,q') = \Tr[(W_q^a \otimes W_{q'}^{a'} \otimes \id) \rho_{S_1S_2S_3}].
\end{equation}
Using the probability distribution $p_{quant}(a,a'|q,q')$, we define the following probability distribution.
\begin{equation}
\pi := \sum_{q,q'}\mathcal{D}_G(q,q') \sum_{a,a'}\V(a,a',q,q')p_{quant}(a,a'|q,q'),
\end{equation}
where the $\V$ function follows definition of Equation \eqref{eq:ver_3col}. Note that, in $\V'$ we perform more checks on the answers from $P^*_3$, and $p_{quant}(a,a'|q,q')$ is a marginal distribution. Therefore we get,

\begin{equation}
\label{eq:pi_low}
\pi \geq \Pr\big[\langle (P^*_1,P^*_2,P^*_3)(G,\rho_{S_1S_2S_3}), V\rangle = 1\big].
\end{equation}
Now, we consider another distribution $\pi_1$, that we define as follows.

\begin{equation}
\label{eq:pi1}
\pi_1 := \sum_{q}\mathcal{D}_G(q) \pi_1(q),
\end{equation}
where $\mathcal{D}_G(q)$ is the marginal of the distribution $\mathcal{D}_G(q,q')$, and 
\begin{equation}
\label{eq:pi1q}
\pi_1(q):=\sum_{a,a}p_{quant}(a,a|q,q).
\end{equation}
Note that, in $\pi_1$, we are only interested in checking the consistency test. Therefore, 
\begin{equation}
\label{eq:pi1_low}
\pi_1 \geq \pi \geq \Pr\big[\langle (P^*_1,P^*_2,P^*_3)(G,\rho_{S_1S_2S_3}),V\rangle = 1\big].
\end{equation}
In the rest of the proof, we aim to get an upper bound on the following quantity. 

\begin{align}\nonumber
|\pi - &\Pr_{f_1,f_2}\big[\langle (\tP_1(G,f_1),\tP_2(G,f_2)),V\rangle = 1\big]| \\ \label{eq:upp_des1}
&\leq \sum_{q,q'}\mathcal{D}_G(q,q') \hspace{-0.3in}\sum_{\substack {a,a':\\\V(a,a',q,q')=1}}|p_{quant}(a,a'|q,q') - p_c(a,a'|q,q')|.
\end{align}
We get the upper bound from the triangle inequality. 

From the Extractor Protocol \ref{prot:3_ext}, we get the distribution $D(f_1;f_2)$ in terms of the measurement strategies $W$, and $|\Psi\rangle_{S_1S_2S_3}$,
\begin{equation}\label{eq:dist_w}
D(f_1;f_2) = \Bigg\lVert \left(\prod_{i=1}^N W^{a_i}_{q_i}\right) \otimes \left(\prod_{i=1}^N W^{a'_i}_{q'_i}\right) \otimes \id |\Psi\rangle_{S_1S_2S_3}\Bigg\rVert_1^2.
\end{equation}

We define another super operator $\mathcal{W}_q$ corresponding to the projector $W_q$ as $\mathcal{W}_q(\sigma) := \sum_{a}W_q^a \sigma W_q^a$, where $\sigma$ is any quantum state. In the next claim, we show that, if the malicious entangled provers convince the verifier in the ZKP Protocol \ref{prot:3col} with high probability, then in our Extractor Protocol \ref{prot:3_ext}, the post-measurement state shared between $P_1^*,P_2^*$ after $P_1^*$ applied the projector for producing a response corresponding to a question $q$ is very close to the original state. Suppose, at the beginning of the protocol the shared state between $P_1^*$, and $P_2^*$ is $\rho_{S_1S_2}:= \Tr_{S_3}[\rho_{S_1S_2S_3}]$. This claim is same as Claim $20$ in \cite{kempe2011entangled}, with a tighter bound. 

\begin{claim}[1]
\label{cl:1}
$\lVert \mathcal{W}_q\otimes \id(\rho_{S_1S_2}) - \rho_{S_1S_2}\rVert_1 \leq 4\sqrt{1 - \pi_1(q)}.$
\end{claim}

\begin{proof}
From definition of $\mathcal{W}_q\otimes \id(\rho_{S_1S_2})$ we get
\begin{align}
\left(\mathcal{W}_q\otimes \id\right)(\rho_{S_1S_2}) &= \sum_{a}\left(W_q^a \otimes \id\right) \rho_{S_1S_2}\left(W_q^a \otimes \id\right) \\
&= \Tr_{S_{3}}\left[\sum_{a}(W_q^a \otimes \id \otimes \id) \rho_{S_1S_2S_3}(W_q^a \otimes \id \otimes \id)\right].
\end{align}
The second equality follows from the definition of the partial trace. Similarly, we can write $\rho_{S_1S_2}$ as,
\begin{equation}
\rho_{S_1S_2} = \Tr_{S_3}\left[\id \otimes \id \otimes \mathcal{W}_q(\rho_{S_1S_2S_3})\right].
\end{equation}
Therefore we get,
\begin{align}\label{eq:temP-bound_0}
&\lVert \mathcal{W}_q\otimes \id(\rho_{S_1S_2}) - \rho_{S_1S_2}\rVert_1 \\
&= \Big\lVert \Tr_{S_{3}}\left[\sum_{a}(W_q^a \otimes \id \otimes \id) \rho_{S_1S_2S_3}(W_q^a \otimes \id \otimes \id)\right] - \Tr_{S_3}\left[\id \otimes \id \otimes \mathcal{W}_q(\rho_{S_1S_2S_3})\right]\Big\rVert_1\\\nonumber
&~~\text{Due to the monotonocity of trace distance we get}\\
&\leq\Big\lVert \sum_{a}(W_q^a \otimes \id \otimes \id) \rho_{S_1S_2S_3}(W_q^a \otimes \id \otimes \id) - \id \otimes \id \otimes \mathcal{W}_q(\rho_{S_1S_2S_3})\Big\rVert_1 \\\nonumber
&\text{Due to the triangle inequality we get,}\\\nonumber
&\leq \Big\lVert (\mathcal{W}_q \otimes \id) (\rho_{S_1S_2S_3})(\mathcal{W}_q \otimes \id) - \sum_{a}(W_q^a \otimes \id \otimes W_q^a) \rho_{S_1S_2S_3}(W_q^a \otimes \id \otimes W_q^a)\Big\rVert_1\\
&+ \Big\lVert \id \otimes \id \otimes \mathcal{W}_q(\rho_{S_1S_2S_3}) - \sum_{a}(W_q^a \otimes \id \otimes W_q^a) \rho_{S_1S_2S_3}(W_q^a \otimes \id \otimes W_q^a)\Big\rVert_1\\\nonumber
&\text{By the symmetry of the state }\rho_{S_1S_2S_3} \text{ we get}\\\label{eq:temp_bound}
& = 2\Big\lVert \id \otimes \id \otimes \mathcal{W}_q(\rho_{S_1S_2S_3}) - \sum_{a}(W_q^a \otimes \id \otimes W_q^a) \rho_{S_1S_2S_3}(W_q^a \otimes \id \otimes W_q^a)\Big\rVert_1.
\end{align}
\end{proof}
From the gentle measurement lemma \cite{winter1999coding}, we get that for any quantum state $\rho$, and any measurement operator (in our case projector) $X\geq 0$,
\begin{equation}
\lVert \sigma - \sqrt{X}\sigma \sqrt{X} \rVert_1 \leq 2\sqrt{1 - \Tr[X\rho]}.
\end{equation}
We apply the gentle measurement lemma in Equation \eqref{eq:temp_bound}, by considering $\sigma := \oplus_{a} (W_q^a \otimes \id \otimes \id) \rho_{S_1S_2S_3}(W_q^a \otimes \id \otimes \id)$, and $X := \oplus_{a} \id \otimes \id \otimes W_q^a$ and get the following upper bound on Equation \eqref{eq:temp_bound}.
\begin{align}
\lVert \mathcal{W}_q\otimes \id(\rho_{S_1S_2}) - \rho_{S_1S_2}\rVert_1 &\leq 4\sqrt{1-\sum_a \Tr[(W_q^a \otimes \id \otimes W_q^a)\rho_{S_1S_2S_3}]}\\
&\leq 4\sqrt{1 - \pi_1(q)}~~~\text{This follows from Equation \eqref{eq:pi1q}}.
\end{align}
This concludes the proof of the claim. \qed
In Extractor \ref{prot:3_ext}, suppose $\rho_{S_{1}S_{2}}(i,j)$ denotes the marginal post-measurement state after performing rewinding at Step $1$ for $i-1$ time and performing rewinding at Step $2$ for $j-1$ time. More precisely, it is defined as follows.

\begin{equation}
\label{rhoij}
\rho_{S_1S_2}(i,j) := \left(\mathcal{W}_{i-1}~\circ\ldots \circ~\mathcal{W}_1\right) \otimes \left(\mathcal{W}_{i-1}~\circ\ldots \circ~\mathcal{W}_1\right) (\rho_{S_1S_2}) 
\end{equation}

By combining the expression of $p_c(a,a'|q,q')$ from Equation \eqref{eq:pclass}, and the expression of $D(f_1;f_2)$ from Equation \eqref{eq:dist_w}, we get,

\begin{align}
p_c(a_i,a'_j|q_i,q'_j) &= \sum_{\substack{f_1 \setminus \{a_i\} \\ f_2\setminus \{a'_i\}}} \Bigg\lVert \left(\prod_{i=1}^N W^{a_i}_{q_i}\right) \otimes \left(\prod_{i=1}^N W^{a'_i}_{q'_i}\right) \otimes \id |\Psi\rangle_{S_1S_2S_3}\Bigg\rVert_1^2\\
&= \Tr[(W_q^a \otimes W_{q'}^{a'}) \rho_{S_1S_2}(i,j)].
\end{align}

Note that, $p_{quant}(a,a'|q,q') = \Tr[(W_q^a \otimes W_{q'}^{a'}) \rho_{S_1S_2}]$. Therefore for a fixed $qi,q'_j$, 

\begin{align}
\label{eq:diff_up}
\sum_{\substack{a_i,a'_j:\\\V(a_i,a'_j,q_i,q'_j)}}|p_{quant}(a_i,a'_j|q_i,q'_j) - p_{c}(a_i,a'_j|q_i,q'_j)| \leq \lVert\rho_{S_1S_2}(i,j) - \rho_{S_1S_2}\rVert_1.
\end{align}In the extractor, suppose in Step $1$, after the $i$-th step of the rewinding the shared state is $\rho_{S_{1}S_{2}S_{3}}(i)$. 

In the following lemma, we give an upper bound on the trace distance between $\rho_{S_{1}S_{2}S_{3}}(i,j)$ and $\rho_{S_{1}S_{2}S_{3}}$. 

\begin{lemma}
\label{lem:up_td}
$\lVert\rho_{S_{1}S_{2}}(i,j) - \rho_{S_{1}S_{2}}\rVert_1 \leq 4\left[\sum_{i' = 1}^{i-1}\sqrt{1 - \pi_1(q_{i'})} + \sum_{j' = 1}^{j-1}\sqrt{1 - \pi_1(q_{j'})}\right]$.
\end{lemma}

\begin{proof}
We follow the steps of the proof of Claim $21$ in \cite{kempe2011entangled} for proving this lemma. We prove this claim by the principle of mathematical induction. Note that, by definition of the state $\rho_{S_1S_2}(i,j)$, we get $\rho_{S_1S_2}(1,1) = \rho_{S_1S_2}$. Suppose the statement is true for some $(i=\tilde i,j=\tilde j)$. Therefore,
\begin{equation}
\label{eq:ind_hyp}
\lVert\rho_{S_{1}S_{2}}(\tilde i,\tilde j) - \rho_{S_{1}S_{2}}\rVert_1 \leq 4\left[\sum_{i' = 1}^{\tilde i-1}\sqrt{1 - \pi_1(q_{i'})} + \sum_{j' = 1}^{\tilde j-1}\sqrt{1 - \pi_1(q_{j'})}\right]
\end{equation}
So for $i = \tilde i +1$, by triangle inequality, we get
\begin{align}
&\lVert\rho_{S_{1}S_{2}}(\tilde i +1,\tilde j) - \rho_{S_{1}S_{2}}\rVert_1 \\ \nonumber 
&\leq \lVert\rho_{S_{1}S_{2}}(\tilde i +1,\tilde j) - (\mathcal{W}_{q_{\tilde i}}\otimes \id )(\rho_{S_{1}S_{2}})\rVert_1 + \lVert\rho_{S_{1}S_{2}} - (\mathcal{W}_{q_{\tilde i}}\otimes \id )(\rho_{S_{1}S_{2}})\rVert_1.
\end{align}
From the definition of $\rho_{S_1S_2}(i,j)$ and from Claim $(1)$:
\begin{align}
\|\cdot \|_1 \leq \lVert (\mathcal{W}_{q_{\tilde i}}\otimes \id) (\rho_{S_{1}S_{2}}(\tilde i,\tilde j) )- (\mathcal{W}_{q_{\tilde i}}\otimes \id )(\rho_{S_{1}S_{2}})\rVert_1 + 4\sqrt{1 - \pi_1(q_{\tilde i})}.
\end{align}
Then from the monotonocity of the trace distance
\begin{align}
\|\cdot \|_1  \leq \lVert \rho_{S_{1}S_{2}}(\tilde i,\tilde j) )- \rho_{S_{1}S_{2}}\rVert_1 + 4\sqrt{1 - \pi_1(q_{\tilde i})}
\end{align}
Finally from the induction hypothesis, i.e., Equation \eqref{eq:ind_hyp}, we get 
\begin{align}
\|\cdot \|_1  \leq 4\left[\sum_{i' = 1}^{\tilde i}\sqrt{1 - \pi_1(q_{i'})} + \sum_{j' = 1}^{\tilde j-1}\sqrt{1 - \pi_1(q_{j'})}\right].
\end{align}
Therefore, by the principle of mathematical induction we show that the statement of this lemma is true for $i,j \geq 1$. This concludes the proof this lemma.
\qed
\end{proof}
We now revisit the expression of $\left|\pi- \Pr_{f_1,f_2}\big[\langle (\tP_1(G,f_1),\tP_2(G,f_2)),V\rangle = 1\big]\right|$ from \Cref{eq:upp_des1}:  
\begin{align*}
    &\left|\pi- \Pr_{f_1,f_2}\big[\langle (\tP_1(G,f_1),\tP_2(G,f_2)),V\rangle = 1\big]\right|\\
    &\leq \sum_{q,q'}\mathcal{D}_G(q,q') \sum_{\substack {a,a':\\\V(a,a',q,q')=1}}\left|p_{quant}(a,a'|q,q') - p_c(a,a'|q,q')\right|.
\end{align*}
Substituting $|p_{quant}(a,a'|q,q') - p_c(a,a'|q,q')|$ from \Cref{eq:diff_up}, we get
\begin{align*}
    |\cdot | \leq \sum_{q,q'}\mathcal{D}_G(q,q') \lVert\rho_{S_1S_2}(i,j) - \rho_{S_1S_2}\rVert_1
\end{align*}
Then substituting $\lVert\rho_{S_1S_2}(i,j) - \rho_{S_1S_2}\rVert_1$ from \Cref{lem:up_td}, we get
\begin{align*}
|\cdot|&\leq 4\sum_{i=1}^N \sum_{j=1}^N \mathcal{D}_G(q_i,q'_j)\left[\sum_{i' = 1}^{i-1}\sqrt{1 - \pi_1(q_{i'})} + \sum_{j' = 1}^{j-1}\sqrt{1 - \pi_1(q_{j'})}\right]\\
&\leq 8 \sum_{i=1}^N \sum_{j=1}^N\sum_{i'=1}^{i-1} \mathcal{D}_G(q_i,q'_{j})\left[\sqrt{1 - \pi_1(q_{i'})}\right]\\
&\leq 8\sum_{i=1}^N \sum_{i'=1}^{i-1}\mathcal{D}_G(q_i)\left[\sqrt{1 - \pi_1(q_{i'})}\right]\\
& \leq 8\sum_{i=1}^N \sqrt{1 - \pi_1(q_i)} \sum_{i'=i+1}^N \mathcal{D}_G
(q_{i'}).
\end{align*}
Assuming $\mathcal{D}_G(q_1) \geq ... \geq \mathcal{D}_G(q_N)$, we get 
\begin{align*}
|\cdot|&\leq 8\sum_{i=1}^N (N - i)\mathcal{D}_G(q_i) \sqrt{1 - \pi_1(q_i)}\\
&\leq 8N\sum_{i=1}^N \mathcal{D}_G(q_{i})\left[\sqrt{1 - \pi_1(q_{i})}\right].
\end{align*}
From \Cref{eq:pi1}, and since the function $\sqrt{1 - \pi_1(q_{i'})}$ is concave, we have 
\begin{align*}
|\cdot | &\leq 8N\sqrt{1-\pi_1}\\
&\leq 8N\sqrt{1-\pi}\\
& \leq 8N\sqrt{1 - \Pr\big[\langle (P^*_1,P^*_2,P^*_3)(G,\rho_{S_1S_2S_3}), V\rangle = 1\big]}.
\end{align*}
where the second last ineqaulity comes from the fact that $\pi_1 \geq \pi$ and the last inequality comes from \Cref{eq:pi1_low}.
Therefore, the above series of inequality gives us the following relation.
\begin{align}\label{eq:pi_final2}
&\left|\pi-\Pr_{f_1,f_2}\big[\langle (\tP_1(G,f_1),\tP_2(G,f_2)),V\rangle = 1\big]\right|\leq 8N\sqrt{1 - \Pr\big[\langle (P^*_1,P^*_2,P^*_3)(G,\rho_{S_1S_2S_3}), V\rangle = 1\big]}. 
\end{align}

Note that in our case, $\pi \geq \Pr_{f_1,f_2}\big[\langle (\tP_1(G,f_1),\tP_2(G,f_2)),V\rangle = 1\big]$. Therefore, we can write Equation \eqref{eq:pi_final2} as
\begin{align}
\nonumber &\pi- \Pr_{f_1,f_2}\big[\langle (\tP_1(G,f_1),\tP_2(G,f_2)),V\rangle = 1\big] \leq 8N\sqrt{1 - \Pr\big[\langle (P^*_1,P^*_2,P^*_3)(G,\rho_{S_1S_2S_3}), V\rangle = 1\big]}.
\end{align}
This gives 
\begin{align}
&\Pr\big[\langle (P^*_1,P^*_2,P^*_3)(G,\rho_{S_1S_2S_3}), V\rangle = 1\big] - \Pr_{f_1,f_2}\big[\langle (\tP_1(G,f_1),\tP_2(G,f_2)),V\rangle = 1\big] \\ \label{eq:final_dist} 
&\leq 8N\sqrt{1 - \Pr\big[\langle (P^*_1,P^*_2,P^*_3)(G,\rho_{S_1S_2S_3}),V\rangle = 1\big]}\\
&= 16|H|\sqrt{1 - \Pr\big[\langle (P^*_1,P^*_2,P^*_3)(G,\rho_{S_1S_2S_3}), V\rangle = 1\big]},
\end{align}
where the first inequality comes from \Cref{eq:pi1_low}, and the last equality holds as $N= 2|H|$. This concludes the proof. \qed
\end{proof}

We restate and prove \Cref{lem:lower_bound} from the main-text here. 
\begin{lemma}
\label{lem:lower-bound1}
$\delta \geq \frac{16}{2401|H|^4}$, for $|H| \geq 3$.
\end{lemma}

\begin{proof}[Proof of \Cref{lem:lower_bound}]
From \Cref{lma:kappa}, we get the following upper bound.      

\begin{equation}
\label{eq:upp_pclas}
\Pr_{f_1,f_2}\big[\langle (\tP_1(G,f_1),\tP_2(G,f_2)),V\rangle = 1\big]\leq \left(1 - \frac{1}{3|H|}\right).
\end{equation}

From \Cref{lem:dist_q_c}, we get
\begin{align} \label{eq:lem6_rep}
|\Pr\big[&\langle (P^*_1,P^*_2,P^*_3)(G,\rho_{S_1S_2S_3}), V\rangle = 1\big] - \Pr_{f_1,f_2}\big[\langle(\tP_1(G,f_1),\tP_2(G,f_2)),V\rangle = 1\big]| \leq \delta.
\end{align}

Substituting the expression from \Cref{eq:upp_pclas} in \Cref{eq:lem6_rep}, and substituting the value of $\delta$ from \Cref{lem:dist_q_c} we get the following inequality:
\begin{align}
1-p_{quant}^{acc} &+8N\sqrt{1-p_{quant}^{acc}} \geq \frac{1}{3|H|}.
\end{align}
where $\delta':=1-p_{quant}^{acc}=1-\Pr\big[\langle (P^*_1,P^*_2,P^*_3(G,\rho_{S_1S_2S_3})) ,V\rangle = 1\big]$. It follows that
\begin{align}
\sqrt{ \delta'} &\geq \frac{1}{3|H|(8N + \sqrt{\delta'})}\\
&\geq \frac{1}{3|H|(8N + 1)}\\
&= \frac{1}{3|H|(16|H| + 1)}\\
&\geq \frac{1}{49|H|^2},
\end{align}
where the second inequality comes from $\sqrt{1-p_{quant}^{acc}} \leq 1$ and the third equality comes from $N=2|H|$. The last inqueality holds for all graphs with $|H| \geq 3$.
\begin{align}
\delta' \geq \left(\frac{1}{49|H|^2}\right)^2 \Longrightarrow \delta \geq \frac{16}{2401|H|^4},
\end{align}
where the implication is given by definition that $ \delta = 16\tilde \delta'$.
This concludes the proof. \qed

\end{proof}

\section{Appendix: Proofs of New Soundness Error Bounds}
\label{sec:app-soundess}

\subsection{Proof of \Cref{thm:TwoChoiceCoupling}}\label{sec:app-coupling}
\begin{proof}
Let us consider the optimal strategy of the entangled game $G^2$, which can be defined with a pre-shared state $\ket{\psi}$, and projective measurements $\{W_a^x\}_{x,a}$ and $\{W_b^y\}_{b,y}$ used by the two parties respectively.
Note that the state is pure and measurements are projective WLOG since they can be made so with appropriate purification and isometric extension.
We define Bob's reduced state (un-normalised) after Alice's measurement as $\sigma^{xa}=\Tr[W_a^x\dyad{\Psi}]$, and let Bob's response that would be accepted be $V_y(x,a)=\{b:V(a,b,x,y)=1\}$.\\

We can express the winning probability of $G^2$ as
\begin{equation*}
\begin{split}
    \omega^*(G^2)=&\frac{1}{2\abs{X}}\sum_{xya}\sum_{b\in V_y(x,a)}\Tr[(W_a^x\otimes W_b^y)\dyad{\Psi}]\\
    =&\frac{1}{2\abs{X}}\sum_{xya}\sum_{b\in V_y(x,a)}\Tr[W_b^y\sigma^{xa}]\\
    =&\frac{1}{\abs{X}}\sum_{xa}\Tr[\sigma^{xa}]\left\{\frac{1}{2}\sum_{y,b\in V_y(x,a)}\Tr[W_b^y\tilde{\sigma}^{xa}]\right\},
\end{split}
\end{equation*}
where $\tilde{\sigma}^{xa}=\frac{\sigma^{xa}}{\Tr[\sigma^{xa}]}$ is the normalised.
Consider a strategy for the coupling game $G^2_{coup}$ where the same state and measurement are utilised, and Bob performs a second measurement $W_{b'}^{\bar{y}}$ on his state.
As such, the winning probability is bounded
\begin{equation*}
\begin{split}
    \omega^*(G^2_{coup})\geq&\frac{1}{2\abs{X}}\sum_{xya}\sum_{\substack{b\in V_y(x,a)\\ b'\in V_{\bar{y}}(x,a)}}\Tr[(W_a^x\otimes W_{b'}^{\bar{y}}W_b^y)\dyad{\Psi}(W_a^x\otimes W_b^yW_{b'}^{\bar{y}})]\\
    =&\frac{1}{2\abs{X}}\sum_{xya}\sum_{\substack{b\in V_y(x,a)\\ b'\in V_{\bar{y}}(x,a)}}\Tr[W_{b'}^{\bar{y}}W_b^y\sigma^{xa}W_b^y]\\
    =&\frac{1}{\abs{X}}\sum_{xa}\Tr[\sigma^{xa}]\left\{\frac{1}{2}\sum_{\substack{y,b\in V_y(x,a)\\ b'\in V_{\bar{y}}(x,a)}}\Tr[W_{b'}^{\bar{y}}W_b^y\tilde{\sigma}^{xa}W_b^y]\right\}.
\end{split}
\end{equation*}
We note here that $\frac{1}{\abs{X}}\sum_{xa}\Tr[\sigma^{xa}]f(x,a)$ can be treated as an expectation value of $f$ with probability distribution $p_{xa}=\Tr[\sigma^{xa}]$.
The form of $\omega^*(G^2)$ matches $\mathbb{E}_{xa}[F_1]$ while the bound on $\omega^*(G_{coup}^2)$ matches $\mathbb{E}_{xa}[F_2]$ of~\Cref{thm:TwoOperators}.
As such, we can apply the theorem to show
\begin{equation*}
\begin{split}
    \omega^*(G^2_{coup})\geq&\mathbb{E}_{xa}[F_2]\\
    \geq&\mathbb{E}_{xa}\left[\frac{2}{\max_y\abs{V_y(x,a)}}\left(F_1-\frac{1}{2}\right)^2\right]\\
    \geq& \frac{2}{\max_{a,x,y}\abs{V_y(x,a)}}\left(\mathbb{E}_{xa}[F_1]-\frac{1}{2}\right)^2\\
    =&\frac{2}{\abs{S_{max}}}\left(\omega^*(G^2)-\frac{1}{2}\right)^2,
\end{split}
\end{equation*}
where $\abs{S_{max}}:=\max_{a,x,y}\abs{V_y(x,a)}$. The third inequality utilises Jensen's inequality noting that $F_1\geq \frac{1}{2}$.
\end{proof}

\subsection{Proof of \Cref{thm:chai-sound}}
\label{sec:chai-sound}
\begin{proof}
Following the non-signalling argument in~\cite{chailloux2017relativistic}, the winning probability of the coupling game is bounded by $\omega^*(G^2_{coup})\leq\frac{1}{Q}$, based on no signalling.
Moreover, we note that $\abs{S_{max}}=n!$, attributed to the number of possible permutations, with each providing a set containing a single valid response for the challenges.
Combining this with~\Cref{thm:TwoChoiceCoupling}, we arrive at
\begin{equation*}
    \frac{1}{Q}\geq\frac{2}{n!}\left(\omega^*(G^2)-\frac{1}{2}\right)^2,
\end{equation*}
which can be rearranged to give the result.
\end{proof}

\subsection{Proof of \Cref{thm:crep-sound}}
\label{sec:crep-sound}
\begin{proof}
Following the proof in Proposition $4$ in~\cite{crepeau2023zero}, the winning probability of the coupling game is bounded by $\omega^*(G^2_{coup})\leq\frac{1}{Q}$, based on no signalling.
Moreover, we note that $\abs{S_{max}}=2^n$, attributed to the number of possible choices of $z\in\{0,1\}^n$, each providing a single valid set of response for the challenges.
Combining this with~\Cref{thm:TwoChoiceCoupling}, we arrive at
\begin{equation*}
    \frac{1}{Q}\geq\frac{2}{2^n}\left(\omega^*(G^2)-\frac{1}{2}\right)^2,
\end{equation*}
which can be rearranged to give the result.
\end{proof}
\end{document}